\documentclass{llncs}
\usepackage{amssymb, amsmath,graphicx,graphics,enumerate, tkz-graph}

\usepackage{hyperref}
\hypersetup{hypertexnames=false}

\usepackage[T1]{fontenc}
 \usepackage[utf8]{inputenc}
\usepackage{microtype}

\newcommand{\ssi}{\subseteq_i}
\newcommand{\si}{\supseteq_i}

\newcommand{\NP}{{\sf NP}}

\DeclareMathOperator{\cw}{cw}

\newcounter{ctrclaim}[theorem]
\newcounter{ctrobs}[theorem]
\newcounter{ctrcase}[theorem]
\newcounter{ctrsubcase}[ctrcase]
\newcounter{ctrfake}
\renewcommand{\thectrsubcase}{\thectrcase\alph{ctrsubcase}}
\makeatletter
\@addtoreset{ctrclaim}{lemma}
\@addtoreset{ctrcase}{lemma}
\@addtoreset{ctrobs}{lemma}
\makeatother
\makeatother

 \newcommand\displaycase[1]{{\bf #1}}
 \newcommand{\qedllncs}{\qed}
\newenvironment{enumeratei}{\begin{enumerate}[(i)]

}{\end{enumerate}}
\newenvironment{enumerate1}
    {\begin{enumerate}[1.]

    }
    {\end{enumerate}}

\newcommand{\clm}[1]{\setcounter{ctrcase}{0}\medskip\phantomsection\refstepcounter{ctrclaim}\noindent\displaycase{Claim \thectrclaim. }{\em #1}\\}
\newcommand{\smallclm}[1]{\setcounter{ctrcase}{0}\phantomsection\refstepcounter{ctrclaim}\noindent\displaycase{Claim \thectrclaim. }{\em #1}\\}
\newcommand{\lastsmallclm}[1]{\setcounter{ctrcase}{0}\phantomsection\refstepcounter{ctrclaim}\noindent\displaycase{Claim \thectrclaim. }{\em #1}}

\newcommand{\thmcase}[1]{\medskip\phantomsection\refstepcounter{ctrcase}\noindent\displaycase{Case \thectrcase: }{\em #1}\\}
\newcommand{\thmsubcase}[1]{\medskip\phantomsection\refstepcounter{ctrsubcase}\noindent\displaycase{Case \thectrsubcase: }{\em #1}\\}

\spnewtheorem{oproblem}{Open Problem}{\bfseries}{\itshape}

\title{Colouring Diamond-free Graphs\thanks{First and last author supported by EPSRC (EP/K025090/1). An extended abstract of this paper appeared in the proceedings of SWAT 2016~\cite{DDP15}.}}

\author{Konrad K. Dabrowski\inst{1}, Fran\c{c}ois Dross\inst{2} \and Dani\"el Paulusma\inst{1}}
\institute{School of Engineering and Computing Sciences, Durham University,\\ Science Laboratories, South Road, Durham DH1 3LE, United Kingdom
\texttt{\{konrad.dabrowski,daniel.paulusma\}@durham.ac.uk}
\and
Universit\'e de Montpellier, France - Laboratoire d'Informatique, de Robotique et de Micro\'electronique de Montpellier, France
\texttt{francois.dross@ens-lyon.fr}
}

\begin{document}
\maketitle
\setcounter{footnote}{0}

\begin{abstract}
The {\sc Colouring} problem is that of deciding, given a graph~$G$ and an integer~$k$, whether~$G$ admits a (proper) $k$-colouring.
For all graphs~$H$ up to five vertices, we classify the computational complexity of {\sc Colouring} for $(\mbox{diamond},H)$-free graphs. Our proof is based on combining known results together with proving that the clique-width is bounded for $(\mbox{diamond},\allowbreak P_1+\nobreak 2P_2)$-free graphs.
Our technique for handling this case is to reduce the graph under consideration to a $k$-partite graph that has a very specific decomposition. As a by-product of this general technique we are also able to prove boundedness of clique-width for four other new 
classes of $(H_1,H_2)$-free graphs.
As such, our work also continues a recent systematic study into the (un)boundedness of clique-width of $(H_1,H_2)$-free graphs, and our five new classes of bounded clique-width reduce the number of open cases from~13 to~8.
\end{abstract}

\section{Introduction}\label{s-intro}

The {\sc Colouring} problem is that of testing whether a given graph can be coloured with at most~$k$ colours for some given integer~$k$, such that any two adjacent vertices receive different colours. The complexity of {\sc Colouring} is fully understood for general graphs: it is \NP-complete even if $k=3$~\cite{Lo73}.
Therefore it is natural to study its complexity when the input is restricted. A classic result in this area is due to Gr\"otschel, Lov\'asz, and Schrijver~\cite{GLS84}, who proved that {\sc Colouring} is polynomial-time solvable for perfect graphs.

As surveyed in~\cite{C14,DLRR12,GJPS,RS04b}, {\sc Colouring} has been well studied for hereditary graph classes, that is, classes that can be defined by a family~${\cal H}$ of forbidden induced subgraphs. For a family~${\cal H}$ consisting of one single forbidden induced subgraph~$H$, the complexity of {\sc Colouring} is completely classified: the problem is polynomial-time solvable if~$H$ is an induced subgraph of~$P_4$ or $P_1+\nobreak P_3$ and \NP-complete otherwise~\cite{KKTW01}.
Hence, many papers (e.g.~\cite{BGPS12a,DGP14,HL,KKTW01,LM15,Ma13,Ma,Sc05}) have considered the complexity of {\sc Colouring} for 
bigenic hereditary graph classes, that is, graph classes defined by families~${\cal H}$ consisting of two forbidden graphs~$H_1$ and~$H_2$; such classes of graphs are also called $(H_1,H_2)$-free. This classification is far from complete (see~\cite{GJPS} for the state of art).
In fact there are still an infinite number of open cases, including cases where both~$H_1$ and~$H_2$ are small. For instance, Lozin and Malyshev~\cite{LM15} determined the computational complexity of {\sc Colouring} for $(H_1,H_2)$-free graphs for all graphs~$H_1$ and~$H_2$ up to four vertices except when $(H_1,H_2) \in \{(K_{1,3},4P_1),\allowbreak (K_{1,3},2P_1+\nobreak P_2),\allowbreak (C_4,4P_1)\}$
(we refer to Section~\ref{s-prelim} for notation and terminology).

The {\em diamond} is the graph $\overline{2P_1+P_2}$, that is, the graph obtained from the complete graph on four vertices by removing an edge.
Diamond-free graphs are well studied in the literature. For instance, Tucker~\cite{Tu87} gave an $O(kn^2)$ time algorithm for
{\sc Colouring} for perfect diamond-free graphs. 
It is also known that that {\sc Colouring} is polynomial-time solvable for diamond-free graphs
that contain no induced cycle of even length~\cite{KMV09} as well as for
diamond-free graphs that contain no induced cycle of length at least~5~\cite{BGM12}.
Diamond-free graphs also played an important role in proving that the class of $P_6$-free graphs contains 24 minimal obstructions for 4-{\sc Colouring}~\cite{CGSZ15} (that is, the {\sc Colouring} problem for $k=4$).

\subsection{Our Main Result}

In this paper we focus on {\sc Colouring} for $(\mbox{diamond},H)$-free graphs where~$H$ is a graph on at most five vertices. 
It is known that {\sc Colouring} is \NP-complete for $(\mbox{diamond},H)$-free graphs when~$H$ contains a cycle or a claw~\cite{KKTW01} and
polynomial-time solvable for $H=sP_1+\nobreak P_2$ ($s\geq 0$)~\cite{DGP14}, $H=2P_1+\nobreak P_3$~\cite{BDHP15}, $H=P_1+\nobreak P_4$~\cite{BLM04}, $H=P_2+\nobreak P_3$~\cite{DHP0} and $H=P_5$~\cite{AM02}.
Hence, the only graph~$H$ on five vertices that remains is $H=P_1+\nobreak 2P_2$, for which we prove polynomial-time solvability in this paper.
This leads to the following result.

\begin{sloppypar}
\begin{theorem}\label{t-maincolouring}
Let~$H$ be a graph on at most five vertices. Then {\sc Colouring} is polynomial-time solvable for $(\mbox{diamond},H)$-free graphs if~$H$ is a linear forest and
\NP-complete otherwise.
\end{theorem}
\end{sloppypar}
To solve the case $H=P_1+2P_2$, one could try to reduce to a subclass of diamond-free graphs, for which {\sc Colouring} is polynomial-time solvable, such as the aforementioned results of~\cite{BGM12,KMV09,Tu87}. This would require us to deal with the presence of small cycles up to $C_7$, which may not be straightforward. Instead we aim to identify tractability from an underlying property: we show that 
the class of $(\mbox{diamond},\allowbreak P_1+\nobreak 2P_2)$-free graphs has bounded clique-width.
This approach has several advantages and will lead to a number of additional results, as we will discuss in the remainder of Section~\ref{s-intro}.

\medskip
\noindent
{\bf Clique-width} is
a graph decomposition that can be constructed via vertex labels and four specific graph
operations, which ensure that vertices labelled alike will always keep the same
label and thus behave identically. The clique-width of a graph~$G$ is the
minimum number of different labels needed to construct~$G$ using these
four operations (we refer to Section~\ref{s-prelim} for a precise definition).
A graph class~${\cal G}$ has {\it bounded} clique-width if there exists a constant~$c$ such that every graph from~${\cal G}$ has clique-width at most~$c$.

Clique-width is a well-studied graph parameter (see, for instance, the surveys \cite{Gu07,KLM09}). 
An important reason for the popularity
of clique-width is that a number of classes of \NP-complete problems, such as those that are definable in Monadic Second Order Logic using quantifiers on vertices but not on edges,
become polynomial-time solvable on any graph class~${\cal G}$ of bounded clique-width (this follows from
combining results from~\cite{CMR00,EGW01,KR03b,Ra07} with a result from~\cite{Oum08}).
The {\sc Colouring} problem is one of the best-known \NP-complete problems
that is solvable in polynomial time on graph classes of bounded clique-width~\cite{KR03b}; another well-known example of such a problem is
{\sc Hamilton Path}~\cite{EGW01}.

\subsection{Methodology}
The key technique for proving that  $(\mbox{diamond},\allowbreak P_1+\nobreak 2P_2)$-free graphs have bounded clique-width is the use of a certain graph decomposition of $k$-partite graphs.
We  obtain this decomposition by generalizing the so-called canonical decomposition of bipartite graphs, 
which decomposes a bipartite graph into two smaller bipartite graphs
such that edges between these two smaller bipartite graphs behave in a very
restricted way. Fouquet, Giakoumakis and Vanherpe~\cite{FGV99} 
introduced this decomposition and characterized
exactly those bipartite graphs that can recursively be canonically decomposed
into graphs isomorphic to~$K_1$. Such bipartite graphs are said to be totally
decomposable by canonical decomposition. We say that $k$-partite graphs 
are totally $k$-decomposable if they 
can be, according to our generalized definition,
recursively $k$-decomposed into graphs isomorphic to~$K_1$. We show that totally $k$-decomposable graphs have clique-width at most~$2k$.
We prove this result in Section~\ref{s-total}, where we also give a formal definition of canonical decomposition, along with our generalization.

Our goal is to transform $(\mbox{diamond},\allowbreak P_1+\nobreak 2P_2)$-free graphs into
graphs in some class for which we already know that the clique-width is bounded.
Besides the class of totally $k$-decomposable graphs, we will also reduce to other known graph classes of bounded clique-width, such
as the class of $(\mbox{diamond},\allowbreak P_2+\nobreak P_3)$-free graphs~\cite{DHP0} and certain classes of $H$-free bipartite graphs~\cite{DP14}. Of course, our transformations must not change the clique-width by ``too much''. We ensure this by using certain graph operations (described in Section~\ref{s-prelim}) that are known to preserve (un)boundedness of clique-width~\cite{KLM09,LR04}.

\subsection{Consequences for Clique-Width}\label{s-cons}
There are numerous papers  (as listed in, for instance,~\cite{DP15,Gu07,KLM09}) that determine the (un)boundedness of the clique-width 
or variants of it (see e.g.~\cite{BGMS14,HMP12}) of special graph classes. Due
to the complex nature of clique-width, proofs of these results are often long
and technical, and there are still many open cases.
In particular, gaps exist in a number of dichotomies on the (un)boundedness of clique-width for graph classes defined by one or more forbidden induced subgraphs.  As such our paper also continues a line of research~\cite{BDHP15,BDHP15b,DHP0,DP14,DP15} in which we focus on these gaps in a {\it systematic} way. 
It is known~\cite{DP15} that the class of $H$-free graphs has bounded clique-width if and only if~$H$ is an induced subgraph of~$P_4$.
 Over the years 
many partial results~\cite{BL02,BELL06,BKM06,BLM04b,BLM04,BM02,DLRR12,MR99}
on the (un)boundedness of clique-width have appeared for classes of $(H_1,H_2)$-free graphs,
but until recently~\cite{DP15} it was not even known whether the number of missing cases was bounded. 
Combining these older results
with recent progress~\cite{BDHP15,DGP14,DHP0,DP15} reduced the number of open cases to 13 
(up to an equivalence relation)~\cite{DP15}.

As a by-product of our general methodology, we are able not only to settle the case $(H_1,H_2)=(\mbox{diamond},P_1+2P_2)$, but in fact we solve
{\bf five} of the remaining 13 open cases by proving that the class of $(H_1,H_2)$-free graphs has bounded clique-width if 
$$
\begin{array}{llcl}
\mbox{{\bf 1--4:}}\; &H_1=K_3\; \mbox{and}\; H_2\in \{P_1+\nobreak 2P_2,
P_1+\nobreak P_2+\nobreak P_3,
P_1+\nobreak P_5,
S_{1,2,2}\}\; \mbox{or}\\[4pt] 
\mbox{{\bf 5:}} &H_1=\mbox{diamond}\; \mbox{and}\; H_2=P_1+\nobreak 2P_2.\end{array}$$ 
The above graphs are displayed in \figurename~\ref{fig:solved-cases}.
Note that the $(K_3,P_1+\nobreak 2P_2)$-free graph case is 
properly contained in
all four of the other cases.
These four other newly solved cases are pairwise incomparable. 
In Section~\ref{sec:sufficient} we use our key technique on totally $k$-decomposable graphs to find a number of sufficient conditions for a graph class to have bounded clique-width.
We use these conditions in Section~\ref{s-triangle} to prove Results 1--4 and we then prove Result~5 (which relies on Result~1) in Section~\ref{s-diamond}.

\begin{figure}[h]
\begin{center}
\centering
\begin{tabular}{cccccc}
\begin{minipage}{0.08\textwidth}
\begin{center}
\scalebox{0.7}
{\begin{tikzpicture}[scale=1,rotate=90]
\GraphInit[vstyle=Simple]
\SetVertexSimple[MinSize=6pt]
\Vertices{circle}{a,b,c}
\Edges(a,b,c,a)
\end{tikzpicture}}
\end{center}
\end{minipage}
&
\begin{minipage}{0.22\textwidth}
\begin{center}
\scalebox{0.7}{
\begin{tikzpicture}[scale=1]
\GraphInit[vstyle=Simple]
\SetVertexSimple[MinSize=6pt]
\Vertex[x=0,y=0]{a}
\Vertex[x=0,y=0]{a}
\Vertex[a=30,d=2]{b}
\Vertex[a=-30,d=2]{e}
\Vertex[x=3.46410161514,y=0]{d}
\Edges(e,a,b,e,d,b)
\end{tikzpicture}}
\end{center}
\end{minipage}
&
\begin{minipage}{0.15\textwidth}
\begin{center}
\scalebox{0.7}
{\begin{tikzpicture}[scale=1,rotate=90]
\GraphInit[vstyle=Simple]
\SetVertexSimple[MinSize=6pt]
\Vertices{circle}{a,b,c,d,e}
\Edge(b)(c)
\Edge(d)(e)
\end{tikzpicture}}
\end{center}
\end{minipage}
&
\begin{minipage}{0.15\textwidth}
\begin{center}
\scalebox{0.7}
{\begin{tikzpicture}[scale=1]
\GraphInit[vstyle=Simple]
\SetVertexSimple[MinSize=6pt]
\Vertices{circle}{a,b,c,d,e,f}
\Edges(a,b,c)
\Edges(e,f)
\end{tikzpicture}}
\end{center}
\end{minipage}
&
\begin{minipage}{0.15\textwidth}
\begin{center}
\scalebox{0.7}
{\begin{tikzpicture}[scale=1]
\GraphInit[vstyle=Simple]
\SetVertexSimple[MinSize=6pt]
\Vertices{circle}{a,b,c,d,e,f}
\Edges(b,c,d,e,f)
\end{tikzpicture}}
\end{center}
\end{minipage}
&
\begin{minipage}{0.15\textwidth}
\begin{center}
\scalebox{0.7}
{\begin{tikzpicture}[scale=1]
\GraphInit[vstyle=Simple]
\SetVertexSimple[MinSize=6pt]
\Vertex[x=0,y=0]{a}
\Vertex[x=0,y=1]{b}
\Vertex[x=0,y=2]{c}
\Vertex[x=1,y=0]{d}
\Vertex[x=1,y=1]{e}
\Vertex[x=1,y=2]{f}
\Edges(a,b,c)
\Edges(b,e)
\Edges(a,d)
\Edges(c,f)
\end{tikzpicture}}
\end{center}
\end{minipage}\\
\\
$K_3$ & diamond & $P_1+\nobreak 2P_2$ & $P_1+\nobreak P_2+\nobreak P_3$ & $P_1+\nobreak P_5$ & $S_{1,2,2}$
\end{tabular}
\end{center}
\caption{The forbidden graphs considered in this paper.}
\label{fig:solved-cases}
\end{figure}
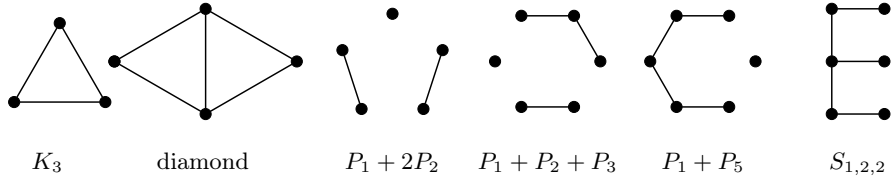

Updating the classification (see~\cite{DP15}) with our five new results gives the following theorem.
Here, ${\cal S}$ is the class of graphs each connected component of which is either a subdivided claw or a path, and 
we write $H\ssi G$ if~$H$ is an induced subgraph of~$G$; see Section~\ref{s-prelim} for notation that we have not formally defined yet.

\newpage
\begin{theorem}\label{thm:classification2}
Let~${\cal G}$ be a class of graphs defined by two forbidden induced subgraphs. Then:
\begin{enumeratei}
\item ${\cal G}$ has bounded clique-width if it is equivalent\footnote{Given four graphs $H_1,H_2,H_3,H_4$, the class of $(H_1,H_2)$-free graphs and the class of $(H_3,H_4)$-free graphs are {\em equivalent} if the unordered pair $H_3,H_4$ can be obtained from the unordered pair $H_1,H_2$ by some combination of the operations (i) complementing both graphs in the pair and (ii) if one of the graphs in the pair is~$K_3$, replacing it with $\overline{P_1+P_3}$ or vice versa.  If two classes are equivalent, then one of them has bounded clique-width if and only if the other one does (see~\cite{DP15}).}
to a class of $(H_1,H_2)$-free graphs such that one of the following holds:
\begin{enumerate1}
\item \label{thm:classification2:bdd:P4} $H_1$ or $H_2 \ssi P_4$;
\item \label{thm:classification2:bdd:ramsey} $H_1=sP_1$ and $H_2=K_t$ for some $s,t$;
\item \label{thm:classification2:bdd:P_1+P_3} $H_1 \ssi P_1+\nobreak P_3$ and $\overline{H_2} \ssi K_{1,3}+\nobreak 3P_1,\; K_{1,3}+\nobreak P_2,\;\allowbreak P_1+\nobreak P_2+\nobreak P_3,\;\allowbreak P_1+\nobreak P_5,\;\allowbreak P_1+\nobreak S_{1,1,2},\;\allowbreak P_6,\; \allowbreak S_{1,1,3}$ or~$S_{1,2,2}$;
\item \label{thm:classification2:bdd:2P_1+P_2} $H_1 \ssi 2P_1+\nobreak P_2$ and $\overline{H_2}\ssi P_1+\nobreak 2P_2,\; 2P_1+\nobreak P_3,\; 3P_1+\nobreak P_2$ or~$P_2+\nobreak P_3$;
\item \label{thm:classification2:bdd:P_1+P_4} $H_1 \subseteq_i P_1+\nobreak P_4$ and $\overline{H_2} \ssi P_1+\nobreak P_4$ or~$P_5$;
\item \label{thm:classification2:bdd:4P_1} $H_1 \subseteq_i 4P_1$ and $\overline{H_2} \ssi 2P_1+\nobreak P_3$;
\item \label{thm:classification2:bdd:K_13} $H_1,\overline{H_2} \ssi K_{1,3}$.
\end{enumerate1}
\item ${\cal G}$ has unbounded clique-width if it is equivalent to a class of $(H_1,H_2)$-free graphs such that one of the following holds:
\begin{enumerate1}
\item \label{thm:classification2:unbdd:not-in-S} $H_1\not\in {\cal S}$ and $H_2 \not \in {\cal S}$;
\item \label{thm:classification2:unbdd:not-in-co-S} $\overline{H_1}\notin {\cal S}$ and $\overline{H_2} \not \in {\cal S}$;
\item \label{thm:classification2:unbdd:K_13or2P_2} $H_1 \si K_{1,3}$ or~$2P_2$ and $\overline{H_2} \si 4P_1$ or~$2P_2$;
\item \label{thm:classification2:unbdd:2P_1+P_2} $H_1 \si 2P_1+\nobreak P_2$ and $\overline{H_2} \si K_{1,3},\; 5P_1,\; P_2+\nobreak P_4$ or~$P_6$;
\item \label{thm:classification2:unbdd:3P_1} $H_1 \si 3P_1$ and $\overline{H_2} \si 2P_1+\nobreak 2P_2,\; 2P_1+\nobreak P_4,\; 4P_1+\nobreak P_2,\; 3P_2$ or~$2P_3$;
\item \label{thm:classification2:unbdd:4P_1} $H_1 \si 4P_1$ and $\overline{H_2} \si P_1 +\nobreak P_4$ or~$3P_1+\nobreak P_2$.
\end{enumerate1}
\end{enumeratei}
\end{theorem}

\subsection{Future Work}

Naturally we would like to extend Theorem~\ref{t-maincolouring} and solve the following open problem.

\begin{oproblem}\label{o-first}
What is the computational complexity of the {\sc Colouring} problem for $(\mbox{diamond},H)$-free graphs when $H$ is a graph on at least six vertices?
\end{oproblem}
Solving Open Problem~\ref{o-first} is highly non-trivial.
It is known that {\sc 4-Colouring} is \NP-complete for $(C_3,P_{22})$-free graphs~\cite{HJP14}. Hence, the polynomial-time results in Theorem~\ref{t-maincolouring} cannot be extended to all linear forests. The first open case to consider would be $H=P_6$, for which only partial results are known. Indeed, {\sc Colouring} is polynomial-time solvable for $(C_3,P_6)$-free graphs~\cite{BKM06}, but its complexity is unknown for $(C_3,P_7)$-free graphs (on a side note, a recent result for the latter graph class is that {\sc 3-Colouring} is polynomial-time solvable~\cite{BCMSZ}). 

\medskip
\noindent
We observe that boundedness of the clique-width of $(\mbox{diamond},P_1+\nobreak 2P_2)$-free graphs implies boundedness of the clique-width of
$(2P_1+\nobreak P_2,\overline{P_1+2P_2})$-free graphs (recall that the diamond is the complement of the graph $2P_1+\nobreak P_2$).
Hence our results imply that {\sc Colouring} can also be solved in polynomial time for graphs in this class.
After incorporating the consequences of our new results and this additional observation, there are 13 classes of $(H_1,H_2)$-free graphs for which {\sc Colouring} could
potentially
still be solved in polynomial time by showing that their clique-width is bounded (see also~\cite{GJPS}):

\begin{oproblem}
Is {\sc Colouring} polynomial-time solvable for $(H_1,H_2)$-free graphs when:
\begin{enumerate}
\item $\overline{H_1}\in \{3P_1,P_1+P_3\}$ and $H_2\in \{P_1+S_{1,1,3},S_{1,2,3}\}$;\\[-10pt]
\item $H_1=2P_1+\nobreak P_2$ and $\overline{H_2} \in \{P_1+\nobreak P_2+\nobreak P_3,\allowbreak P_1+\nobreak P_5\}$;\\[-10pt]
\item $H_1=\mbox{diamond}$
and $H_2 \in \{P_1+\nobreak P_2+\nobreak P_3,\allowbreak P_1+\nobreak P_5\}$;\\[-10pt]
\item $H_1=P_1+P_4$ and $\overline{H_2} \in \{P_1+2P_2,P_2+P_3\}$;\\[-10pt]
\item $\overline{H_1}=P_1+P_4$ and ${H_2} \in \{P_1+2P_2,P_2+P_3\}$;\\[-10pt]
\item $H_1=\overline{H_2}=2P_1+P_3$.
\end{enumerate}
\end{oproblem}

\noindent
As mentioned in Section~\ref{s-cons}, after updating the list of remaining open cases
for clique-width
from~\cite{DP15}, we find that 
eight non-equivalent open cases remain for clique-width. 
These are the following cases.

\begin{oproblem}\label{oprob:twographs}
Does the class of $(H_1,H_2)$-free graphs have bounded or unbounded clique-width when:
\begin{enumerate}
\item \label{oprob:twographs:3P_1} $H_1=3P_1$ and $\overline{H_2} \in \{P_1+\nobreak S_{1,1,3},\allowbreak P_2+\nobreak P_4,\allowbreak S_{1,2,3}\}$;
\item\label{oprob:twographs:2P_1+P_2} $H_1=2P_1+\nobreak P_2$ and $\overline{H_2} \in \{P_1+\nobreak P_2+\nobreak P_3,\allowbreak P_1+\nobreak P_5\}$;
\item \label{oprob:twographs:P_1+P_4} $H_1=P_1+\nobreak P_4$ and $\overline{H_2} \in \{P_1+\nobreak 2P_2,\allowbreak P_2+\nobreak P_3\}$ or
\item \label{oprob:twographs:2P_1+P_3} $H_1=\overline{H_2}=2P_1+\nobreak P_3$.
\end{enumerate}
\end{oproblem}

Bonomo, Grippo, Milani\v{c} and Safe~\cite{BGMS14} determined all pairs  
of connected graphs $H_1,H_2$ 
for which the class of $(H_1,H_2)$-free graphs has power-bounded clique-width.
In order to compare their result with our results for clique-width, we 
would only need to solve the single open case $(H_1,H_2)=(K_3,S_{1,2,3})$,
which is equivalent to the (open) case $(H_1,H_2)=(3P_1,\overline{S_{1,2,3}})$ mentioned in Open Problem~\ref{oprob:twographs}.
This follows because our new result for 
the case $(H_1,H_2)=(K_3,S_{1,2,2})$
has reduced the number of
open cases $(H_1,H_2)$ with $H_1,H_2$ both connected from two to one. 

\section{Preliminaries}\label{s-prelim}

Throughout our paper we only consider finite, undirected graphs without multiple edges or self-loops.
Below we define further graph terminology.

The {\em disjoint union} $(V(G)\cup V(H), E(G)\cup E(H))$ of two vertex-disjoint graphs~$G$ and~$H$ is denoted by~$G+\nobreak H$ and the disjoint union of~$r$ copies of a graph~$G$ is denoted by~$rG$. The {\em complement} of a graph~$G$, denoted by~$\overline{G}$, has vertex set $V(\overline{G})=\nobreak V(G)$ and an edge between two distinct vertices
if and only if these vertices are not adjacent in~$G$. 
For a subset $S\subseteq V(G)$, we let~$G[S]$ denote the subgraph of~$G$ {\it induced} by~$S$, which has vertex set~$S$ and edge set $\{uv\; |\; u,v\in S, uv\in E(G)\}$.
If $S=\{s_1,\ldots,s_r\}$ then, to simplify notation, we may also write $G[s_1,\ldots,s_r]$ instead of $G[\{s_1,\ldots,s_r\}]$.
We use $G \setminus S$ to denote the graph obtained from~$G$ by deleting every vertex in~$S$, i.e. $G \setminus S = G[V(G)\setminus S]$.
We write $H\ssi G$ to indicate that~$H$ is an induced subgraph of~$G$.

The graphs $C_r,K_r,K_{1,r-1}$ and~$P_r$ denote the cycle, complete graph, star and path on~$r$ vertices, respectively.
The graph~$K_{1,3}$ is also called the {\em claw}.
The graph~$S_{h,i,j}$, for $1\leq h\leq i\leq j$, denotes the {\em subdivided claw}, that is,
the tree that has only one vertex~$x$ of degree~$3$ and exactly three leaves, which are of distance~$h$,~$i$ and~$j$ from~$x$, respectively.
Observe that $S_{1,1,1}=K_{1,3}$.
The graph~$S_{1,2,2}$ is also known as the E, 
since it can be drawn like a capital letter~$E$ (see \figurename~\ref{fig:solved-cases}).
Recall that the graph $\overline{2P_1+P_2}$ is known as the {\em diamond}.
The graphs~$K_3$ and $\overline{P_1+2P_2}$ are also known as the {\em triangle}
and the 5-vertex {\em wheel}, respectively.
For a set of graphs $\{H_1,\ldots,H_p\}$, a graph~$G$ is {\em $(H_1,\ldots,H_p)$-free} if it has no induced subgraph isomorphic to a graph in $\{H_1,\ldots,H_p\}$;
if~$p=1$, we may write $H_1$-free instead of $(H_1)$-free.

Let~$X$ be a set of vertices in a graph $G=(V,E)$.
A vertex $y\in V\setminus X$ is {\em complete} to~$X$ if it is adjacent to every vertex of~$X$
and {\em anti-complete} to~$X$ if it is non-adjacent to every vertex of~$X$.
Similarly, a set of vertices $Y\subseteq V\setminus X$ is {\em complete} ({\em anti-complete}) to~$X$ if every vertex in~$Y$ is complete (anti-complete) to~$X$.
A vertex~$y$ or a set~$Y$ is {\em trivial} to~$X$ if it is either complete or anti-complete to~$X$.
Note that if~$Y$ contains both vertices complete to~$X$ and vertices  not complete to~$X$, we may have a situation in which every vertex in~$Y$ is trivial to~$X$, but~$Y$ itself is not trivial to~$X$.

For a graph $G=(V,E)$, the set $N(u)=\{v\in V\; |\; uv\in E\}$ denotes the neighbourhood of $u\in V$.
Let~$X$ and~$Y$ be disjoint sets of vertices in a graph $G=(V,E)$.
If every vertex of~$X$ has at most one neighbour in~$Y$ and vice versa then we say that the edges between~$X$ and~$Y$ form a {\em matching}.
If every vertex of~$X$ has exactly one neighbour in~$Y$ and vice versa then we say that the edges between~$X$ and~$Y$ form a {\em perfect matching}.

A graph is {\it $k$-partite} if its vertex set can be partitioned into~$k$ independent sets (some of which may be empty).
A graph is {\em bipartite} if it is $2$-partite. 
A graph is {\em complete bipartite} if its vertex set can be partitioned into two independent sets that are complete to each other.
For integers $r,s\geq 0$, 
the {\em biclique}~$K_{r,s}$ is the complete bipartite graph with sets in the partition of size~$r$ and~$s$ respectively.
The {\em bipartite complement} of a bipartite graph~$G$ with bipartition~$(X,Y)$ is the graph obtained from~$G$ by replacing every edge from a vertex in~$X$ to a vertex in~$Y$ by a non-edge and vice versa.

\medskip
\noindent
{\bf Clique-Width.}
The {\em clique-width} of a graph~$G$, denoted~$\cw(G)$, is the minimum
number of labels needed to
construct~$G$ by
using the following four operations:
\begin{enumerate}
\item creating a new graph consisting of a single vertex~$v$ with label~$i$;
\item taking the disjoint union of two labelled graphs~$G_1$ and~$G_2$;
\item joining each vertex with label~$i$ to each vertex with label~$j$ ($i\neq j$);
\item renaming label~$i$ to~$j$.
\end{enumerate}

\noindent
An algebraic term that represents such a construction of~$G$ and uses at most~$k$ labels is said to be a {\em $k$-expression} of~$G$ (i.e. the clique-width of~$G$ is the minimum~$k$ for which~$G$ has a $k$-expression).
Recall that a class of graphs~${\cal G}$ has bounded clique-width if there is a constant~$c$ such that the clique-width of every graph in~${\cal G}$ is at most~$c$; otherwise the clique-width of~${\cal G}$ is {\em unbounded}.

Let~$G$ be a graph.
We define the following operations.
For an induced subgraph $G'\ssi G$, the {\em subgraph complementation} operation (acting on~$G$ with respect to~$G'$) replaces every edge present in~$G'$
by a non-edge, and vice versa.
Similarly, for two disjoint vertex subsets~$S$ and~$T$ in~$G$, the {\em bipartite complementation} operation with respect to~$S$ and~$T$ acts on~$G$ by replacing
every edge with one end-vertex in~$S$ and the other one in~$T$ by a non-edge and vice versa.

We now state some useful facts about how the above operations (and some other ones) influence the clique-width of a graph.
We will use these facts throughout the paper.
Let $k\geq 0$ be a constant and let~$\gamma$ be some graph operation.
We say that a graph class~${\cal G'}$ is {\em $(k,\gamma)$-obtained} from a graph class~${\cal G}$
if the following two conditions hold:
\begin{enumeratei}
\item every graph in~${\cal G'}$ is obtained from a graph in~${\cal G}$ by performing~$\gamma$ at most~$k$ times, and
\item for every $G\in {\cal G}$ there exists at least one graph
in~${\cal G'}$ obtained from~$G$ by performing~$\gamma$ at most~$k$ times.
\end{enumeratei}

\noindent
We say that~$\gamma$ {\em preserves} boundedness of clique-width if
for any finite constant~$k$ and any graph class~${\cal G}$, any graph class~${\cal G}'$ that is $(k,\gamma)$-obtained from~${\cal G}$
has bounded clique-width if and only if~${\cal G}$ has bounded clique-width.

\begin{enumerate}[\bf F{a}ct 1.]
\item \label{fact:del-vert} Vertex deletion preserves boundedness of clique-width~\cite{LR04}.\\[-1em]

\item \label{fact:comp} Subgraph complementation preserves boundedness of clique-width~\cite{KLM09}.\\[-1em]

\item \label{fact:bip} Bipartite complementation preserves boundedness of clique-width~\cite{KLM09}.\\[-1em]

\end{enumerate}
The following lemma is easy to show.
\begin{lemma}\label{lem:atmost-2}
The clique-width of a graph of maximum degree at most~$2$ is at most~$4$.
\end{lemma}
Two vertices are {\em false twins} if they have the same neighbourhood
(note that such vertices must be non-adjacent).
The following lemma follows immediately from the definition of clique-width.
\begin{lemma}\label{lem:false-twin}
If a vertex~$x$ in a graph~$G$ has a false twin then $\cw(G)=\cw(G \setminus \{x\})$.
\end{lemma}
We will also make use of the following two results.

\begin{lemma}[\cite{DHP0}]\label{lem:diamond-P_2+P_3}
The class of $(\mbox{diamond},\allowbreak P_2+\nobreak P_3)$-free graphs has bounded clique-width.
\end{lemma}

\begin{lemma}[\cite{DP14}]\label{lem:bipartite}
Let~$H$ be a graph.
The class of $H$-free bipartite graphs has bounded
clique-width if and only~if
\begin{itemize}
\item [$\bullet$] $H=sP_1$ for some $s\geq 1$;
\item [$\bullet$] $H\ssi K_{1,3}+3P_1$;
\item [$\bullet$] $H\ssi K_{1,3}+P_2$;
\item [$\bullet$] $H\ssi P_1+S_{1,1,3}$ or
\item [$\bullet$] $H\ssi S_{1,2,3}$.
\end{itemize}
\end{lemma}

In some of our proofs we will use the fact that $S_{1,2,3}$-free bipartite graphs have bounded clique-width, which follows from Lemma~\ref{lem:bipartite}.
Alternatively we could have used the result of Lozin~\cite{Lo02}, who showed that $S_{1,2,3}$-free bipartite graphs have clique-width at most~$5$.

\section{Totally $k$-Decomposable Graphs}\label{s-total}

In this section we describe our key technique, which is based on a decomposition of bipartite graphs introduced by
Fouquet, Giakoumakis and Vanherpe~\cite{FGV99}, which is defined as follows.

Let~$G$ be a bipartite graph with a vertex bipartition $(V_1,V_2)$.
A {\em $2$-decomposition of~$G$} with respect to $(V_1,V_2)$ consists of
two {\em non-empty} graphs  $G[V_1' \cup V_2']$ and $G[V_1'' \cup\nobreak V_2'']$ such that:
\newpage
\begin{enumeratei}
\item for $i\in \{1,2\}$, $V_i'\cup V_i''=V_i$ and $V_i'\cap V_i''=\emptyset$; 
\item $V_1'$ is either complete or anti-complete to $V_2''$ in $G$;
\item $V_2'$ is either complete or anti-complete to $V_1''$ in $G$.
\end{enumeratei}
Note that $V_1'\cup V_1''$ and $V_2'\cup V_2''$ are independent sets in~$G$ and that the last two conditions imply that each of $G[V_1' \cup V_2'']$ and $G[V_1'' \cup V_2']$ is either an independent set or a biclique.
Observe that we do not impose restrictions on the bipartite graphs $G'=G[V_1'\cup V_2']$ and $G''=G[V_1''\cup V_2'']$. 
If~$G$ has a $2$-decomposition $G',G''$ with respect to some bipartition, we say that~$G$ can be {\it $2$-decomposed} into~$G'$ and~$G''$.
A graph~$G$ is \emph{totally decomposable by canonical decomposition} if it can be recursively $2$-decomposed into graphs isomorphic to~$K_1$. 
Note that if~$G$ has a $2$-decomposition $G',G''$ with respect to some bipartition $(V_1,V_2)$, this does not force us to  decompose~$G'$ and~$G''$ with respect to a sub-partition of $(V_1,V_2)$. As we will see, this distinction does not make a difference for bipartite graphs, but it will become an issue when we extend the notion to $k$-partite graphs when $k \geq 3$.

Fouquet, Giakoumakis and Vanherpe proved the following characterization, which we will need for our proofs
(see \figurename~\ref{fig:P7-S123} for pictures of~$P_7$ and~$S_{1,2,3}$).

\begin{lemma}[\cite{FGV99}]\label{lem:canon}
A bipartite graph is totally decomposable by canonical decomposition if and only if it is $(P_7,S_{1,2,3})$-free.
\end{lemma}

\begin{figure}[h]
\begin{center}
\begin{tabular}{cc}
\begin{minipage}{0.25\textwidth}
\begin{center}
\scalebox{0.7}{
\begin{tikzpicture}[scale=1]
\GraphInit[vstyle=Simple]
\SetVertexSimple[MinSize=6pt]
\Vertex[x=0,y=0]{x0}
\Vertex[x=1,y=0]{x1}
\Vertex[x=2,y=0]{x2}
\Vertex[x=3,y=0]{x3}
\Vertex[x=0.5,y=1]{y0}
\Vertex[x=1.5,y=1]{y1}
\Vertex[x=2.5,y=1]{y2}
\Edges(x0,y0,x1,y1,x2,y2,x3)
\end{tikzpicture}}
\end{center}
\end{minipage}
&
\begin{minipage}{0.25\textwidth}
\begin{center}
\scalebox{0.7}{
\begin{tikzpicture}[scale=1]
\GraphInit[vstyle=Simple]
\SetVertexSimple[MinSize=6pt]
\Vertex[x=0.5,y=0]{x0}
\Vertex[x=1,y=0]{x1}
\Vertex[x=1.5,y=0]{x2}
\Vertex[x=2.5,y=0]{x4}
\Vertex[x=1,y=1]{z}
\Vertex[x=2,y=1]{x3}
\Vertex[x=0,y=1]{x5}
\Edges(x5,x0,z)
\Edge(x1)(z)
\Edge(x2)(z)
\Edges(x2,x3,x4)
\end{tikzpicture}}
\end{center}
\end{minipage}\\
\\
$P_7$ & $S_{1,2,3}$
\end{tabular}
\end{center}
\caption{The forbidden graphs from Lemma~\ref{lem:canon}.}\label{fig:P7-S123}
\end{figure}
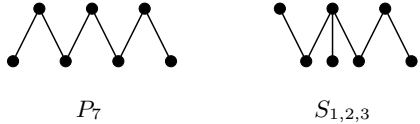

For our purposes we need to generalize the notion of totally decomposable bipartite graphs to $k$-partite graphs for $k\geq 2$, and we will also need to
partially classify graphs with this modified notion, in effect generalizing Lemma~\ref{lem:canon}.

Let~$G$ be a $k$-partite graph with a {\em fixed} vertex $k$-partition $(V_1,\ldots,V_k)$.
A {\em $k$-decomposition of~$G$} with respect to the partition $(V_1,\ldots,V_k)$ 
consists of
two non-empty graphs, each with their own partition:
$G'=G[V_1' \cup \cdots \cup V_k']$ 
with partition $(V_1',V_2',\ldots,V_k')$ 
and
$G''=G[V_1''\cup \cdots \cup V_k'']$ with partition $(V_1'',V_2'',\ldots,V_k'')$, such that:
\begin{enumeratei}
\item for $i\in \{1,\ldots,k\}$, $V_i'\cup V_i''=V_i$ and $V_i'\cap V_i''=\emptyset$;
\item for all $i,j \in \{1,\ldots,k\}$, 
$V_i'$ is either complete or anti-complete to~$V_j''$ in~$G$. 
\end{enumeratei}

\noindent
Note that the last condition holds for $i=j$ by definition, since $V_i=V_i' \cup V_i''$ is an independent set in~$G$.
Also note that in the above definition, $(V_1',V_2',\ldots,V_k')$ and $(V_1'',V_2'',\ldots,V_k'')$ are sub-partitions of $(V_1,V_2,\ldots,V_k)$, in the sense that 
$V_i'=V_i\cap V(G')$ and $V_i''=V_i\cap V(G)$ for $i \in \{1,\ldots,k\}$,
so the original partition on~$G$ uniquely specifies the partitions on~$G'$ and~$G''$.

If a graph~$G$ with a fixed $k$-partition has a $k$-decomposition with respect to this partition into two graphs~$G'$ and~$G''$ (with their associated sub-partitions), we say that~$G$ can be {\em $k$-decomposed} into~$G'$ and~$G''$ (with each of these subgraphs getting the appropriate sub-partition).
We say that~$G$ is {\em totally $k$-decomposable} with respect to some fixed partition if~$G$ can be recursively $k$-decomposed 
{\em with respect to this  fixed partition}
into graphs isomorphic to~$K_1$.
Note that by definition, if a graph~$H$ appears in a total $k$-decomposition of~$G$ with respect to some fixed partition $(V_1,\ldots,V_k)$, then
the $k$-partition $(V^H_1,V^H_2,\ldots,V^H_k)$ of~$H$ used to partition $H$ satisfies $V^H_i =V_i\cap V(H)$ for $i=1,\ldots,k$. 
This property will be necessary for us to be able to use inductive arguments ``safely.''

To compare graphs that are totally decomposable by canonical decomposition and graphs that are totally $2$-decomposable,
we observe that every connected bipartite graph~$G$ has a unique bipartition (up to isomorphism and swapping the two independent sets in the bipartition).
Also, if~$G$ is totally decomposable by canonical decomposition, then this decomposition can recursively be done component-wise.
Hence, in each step of the recursion, we may decompose with respect to an arbitrary bipartition of the graph under consideration. 
This means that the definitions of total canonical decomposability and total $2$-decomposability are equivalent.
However, for $k>2$, a connected graph can have multiple $k$-partitions, even up to isomorphism and permuting the independent sets of the partition.
Therefore, unlike for $k=2$, we need to fix the partition of the subgraphs~$G'$ and~$G''$ in the definition of total $k$-decomposability.

As mentioned, for our proofs we need to generalize Lemma~\ref{lem:canon}. It seems difficult 
to give a full characterization of totally $k$-decomposable graphs for $k\geq 3$. However, the following lemma is sufficient for our purposes.

\begin{lemma}\label{lem:123-imp}
A $3$-partite graph~$G$ is totally $3$-decomposable with respect to a $3$-partition $(V_1,V_2,V_3)$ if the following two conditions are both satisfied:
\begin{itemize}
\item $G[V_1 \cup V_2], G[V_1 \cup V_3]$ and $G[V_2 \cup V_3]$ are all $(P_7,S_{1,2,3})$-free, and
\item for every $v_1 \in V_1$, every $v_2 \in V_2$ and every $v_3 \in V_3$, the graph $G[v_1,v_2,v_3]$ is isomorphic neither to~$K_3$ nor to~$3P_1$.
\end{itemize}
\end{lemma}

\begin{proof}
Let~$G$ be a $3$-partite graph with a $3$-partition $(V_1,V_2,V_3)$ such that both conditions are satisfied.
Note that any induced subgraph~$H$ of~$G$
(with partition $(V(H) \cap V_1, V(H) \cap V_2, V(H) \cap V_3)$) also satisfies the hypotheses of the lemma.
This enables us to apply induction.
It is therefore sufficient to show that~$G$ has a $3$-decomposition with respect to the given $3$-partition.

If~$V_1$ is empty then~$G$ is a $(P_7,S_{1,2,3})$-free bipartite graph and is therefore totally $2$-decomposable 
with respect to the partition $(V_2,V_3)$
by Lemma~\ref{lem:canon}
(and is thus totally $3$-decomposable with respect to the partition $(V_1,V_2,V_3)$).
By symmetry, we may therefore assume that every set~$V_i$ is non-empty.

Now $G[V_1,V_2]$ is a bipartite $(P_7,S_{1,2,3})$-free graph, so by Lemma~\ref{lem:canon}, $G[V_1\cup V_2]$ is totally $2$-decomposable.
Since~$V_1$ and~$V_2$ are both non-empty, it follows that~$V_1$ can be partitioned into two sets~$V_1'$ and~$V_1''$ and~$V_2$ can be partitioned into two sets~$V_2'$ and~$V_2''$, such that~$V_1'$ is either complete or anti-complete to~$V_2''$, and~$V_2'$ is either complete or anti-complete to~$V_1''$. 
Since the graphs $G[V_1' \cup V_2']$ and $G[V_1'' \cup V_2'']$ in this decomposition must be non-empty, it follows that $V_1' \cup V_2'$ and $V_1'' \cup V_2''$ must be non-empty.
Since for $i \in \{1,2\}$ we know that~$V_i=V_i'\cup V_i''$ is non-empty, at least one of~$V_i'$ and~$V_i''$ is non-empty.
Hence, combining these two observations, we may assume without loss of generality that~$V_1'$ and~$V_2''$ are non-empty.
Assume that these sets are maximal, that is, no vertex of~$V_1''$ (respectively $V_2'$) can be moved to~$V_1'$ (respectively~$V_2''$).
Note that~$V_1''$ or $V_2'$ may be empty.

We will prove that we can partition~$V_3$ into sets~$V_3'$ and~$V_3''$, such that for all
$i,j \in \{1,2,3\}$, $V_i'$ is complete or anti-complete to~$V_j''$.
Note that we already know that~$V_1'$ (respectively~$V_2'$) is complete or anti-complete to $V_2''$ (respectively~$V_1''$).
Also note that for $i \in \{1,2,3\}$, $V_i'$ is automatically anti-complete to~$V_i''$, since~$V_i$ is an independent set.

First suppose that~$V_1'$ is complete to~$V_2''$.
If a vertex of~$V_3$ has a neighbour in both~$V_1'$ and~$V_2''$ then these three vertices would form a forbidden~$K_3$, so every vertex in~$V_3$ is anti-complete to~$V_1'$ or~$V_2''$.
Let~$V_3'$ be the set of vertices in~$V_3$ that are anti-complete to~$V_2''$ and let $V_3''=V_3 \setminus V_3'$.
Note that~$V_3''$ must be anti-complete to~$V_1'$.
Suppose, for contradiction, that $z \in V_3'$ has a non-neighbour $v \in V_1''$.
Since~$V_1'$ is maximal, $v$ must have a non-neighbour $w \in V_2''$.
This means that $G[v,w,z]$ is a~$3P_1$.
This contradiction means that~$V_1''$ is complete to~$V_3'$.
Similarly, $V_2'$ is complete to~$V_3''$.
Therefore $G[V_1' \cup V_2' \cup V_3']$ and $G[V_1'' \cup V_2'' \cup V_3'']$ form the required $3$-decomposition of~$G$.

Now suppose that~$V_1'$ is anti-complete to~$V_2''$. 
If a vertex of~$V_3$ has a non-neighbour in both~$V_1'$ and~$V_2''$ then these three vertices would induce a forbidden~$3P_1$, so every vertex in~$V_3$ is complete to~$V_1'$ or~$V_2''$.
Let~$V_3'$ be the set of vertices in~$V_3$ that are complete to~$V_2''$ and let $V_3''=V_3 \setminus V_3'$.
Note that~$V_3''$ must be complete to~$V_1'$.
By using similar arguments to those in the previous case, we find that~$V_1''$ is anti-complete to~$V_3'$ and~$V_2'$ is anti-complete to~$V_3''$.
Hence, $G[V_1' \cup V_2' \cup V_3']$ and $G[V_1'' \cup V_2'' \cup V_3'']$ form the required $3$-decomposition of~$G$.
This completes the proof.
\qedllncs
\end{proof}

We also need the following lemma.

\begin{lemma}\label{lem:k-decomp}
Let~$G$ be a $k$-partite graph with vertex partition $(V_1,\ldots,V_k)$.
If~$G$ is totally $k$-decomposable with respect to this partition, then the clique-width of~$G$ is at most~$2k$.
Moreover, there is a $2k$-expression for~$G$ that assigns, for $i \in \{1,\ldots,k\}$, label~$i$ to every vertex of~$V_i$.
\end{lemma}

\begin{proof}
We prove the lemma by induction on the number of vertices.
If~$G$ contains only one vertex then the lemma holds trivially.
Suppose that the lemma is true for all $k$-partite graphs~$H$ on at most~$n$ vertices and for all $k$-partitions $(V^H_1,\ldots,V^H_2)$ with respect to which~$H$ is totally $k$-decomposable.
Let~$G$ be a graph on $n+1$ vertices that is totally $k$-decomposable with respect to a vertex partition $(V_1,\ldots,V_k)$.
Then, 
we can partition every set~$V_i$ into two sets~$V_i'$ and~$V_i''$
in such a way that
each set~$V_i'$ is either complete or anti-complete to each set~$V_j''$ for all $i,j \in \{1,\ldots,k\}$
and $G'=G[V_1'\cup \ldots \cup V_k']$ and $G''=G[V_1''\cup \ldots \cup V_k'']$ are totally $k$-decomposable with respect to the partitions $(V_1',\ldots,V_k')$ and $(V_1'',\ldots,V_k'')$, respectively.

As both~$G'$ and~$G''$ are smaller graphs that~$G$, we can apply the induction hypothesis.
Hence, we can find a $2k$-expression that constructs~$G'$ such that the vertices in each set~$V_i'$ have label~$i$ for $i \in \{1,\ldots,k\}$.
Similarly, we can find a $2k$-expression that constructs~$G''$ such that the vertices in each set~$V_j''$ have label~$k+j$ for $j \in \{1,\ldots,k\}$.
We take the disjoint union of these two constructions.
Next, for $i,j \in \{1,\ldots,k\}$, we join the vertices with label~$i$ to the vertices with label~$k+j$ if and only if~$V_i'$ is complete to~$V_j''$ in~$G$.
Finally, for $i \in \{1,\ldots,k\}$, we relabel the vertices with label $k+i$ to have label~$i$.
This completes the proof of the lemma.
\qedllncs
\end{proof}

\section{Sufficient Conditions for $(K_3,S_{1,2,3})$-free Graphs}\label{sec:sufficient}

We observe that the classes of $(K_3,P_1+\nobreak 2P_2)$-free, $(K_3,P_1+\nobreak P_2+\nobreak P_3)$-free,
$(K_3,P_1+\nobreak P_5)$-free and $(K_3,S_{1,2,2})$-free graphs are all subclasses of the class of
$(K_3,S_{1,2,3})$-free graphs. In order to prove that each of the four subclasses has bounded clique-width, we investigate, in this section, sufficient conditions for 
a subclass of $(K_3,S_{1,2,3})$-free graphs to be of bounded clique-width.
We present these conditions in the form of two lemmas. The proof of the second lemma uses the results from the previous section. 
We will not use the two new lemmas directly when proving that the class of $(\mbox{diamond},P_1+\nobreak 2P_2)$-free graphs has bounded clique-width. 
However, our proof of that result does rely on these two lemmas indirectly, as it depends on the $(K_3,P_1+\nobreak 2P_2)$-free case.

The first lemma implies that the four triangle-free cases in our new results hold when the graph class under consideration is in addition $C_5$-free.

\newpage
\begin{lemma}\label{lem:noC5}
The class of $(K_3,C_5,S_{1,2,3})$-free graphs has bounded clique-width.
\end{lemma}

\begin{proof}
Let~$G$ be a $(K_3,C_5,S_{1,2,3})$-free graph.
We may assume that $G$ is connected.
If~$G$ is bipartite, then it is an $S_{1,2,3}$-free bipartite graph, so it has bounded clique-width by Lemma~\ref{lem:bipartite}.
We know that~$G$ is $(C_3,C_5)$-free (since $C_3=K_3$).
We may therefore assume that~$G$ contains an induced odd cycle~$C$ on~$k$ vertices, say $v_1-v_2-\cdots-v_k-v_1$, where $k \geq 7$.
Assume that $C$ is an odd cycle of minimum length in~$G$.

Suppose that not every vertex of~$G$ is in~$C$.
Since~$G$ is connected, we may assume that there is a vertex~$v$ not in~$C$ that has a neighbour in~$C$.
Suppose~$v$ is adjacent to precisely one vertex of~$C$.
If~$v$ is adjacent to~$v_3$, but has no other neighbours on~$C$ then $G[v_3,v,v_2,v_1,v_4,v_5,v_6]$ is an~$S_{1,2,3}$, a contradiction.
By symmetry, it follows that~$v$ must be adjacent to at least two vertices of~$C$.
Note that since~$G$ is $K_3$-free, no vertex outside of~$C$ can be adjacent to two consecutive vertices of~$C$.

Suppose that~$v$ is adjacent to~$v_1$ and~$v_i$ and non-adjacent to $v_2,\ldots,v_{i-1}$ for some even~$i$ with $i\leq k-2$.
Then $G[v,v_1,v_2,\ldots,v_i]$ would be an odd cycle on less than~$k$ vertices, contradicting the minimality of~$k$.
By a parity argument, since~$C$ is an odd cycle, it follows that~$v$ must be adjacent to precisely two vertices of~$C$, which must be at distance~$2$ away from each other on the cycle.

Let~$V_i$ be the set of vertices outside of~$C$ that are adjacent to~$v_{i-1}$ and~$v_{i+1}$ (subscripts interpreted modulo~$k$) and let~$U$ be the set of vertices that have no neighbour in~$C$.
Suppose, for contradiction, that~$U$ is non-empty.
Since~$G$ is connected, without loss of generality there is a vertex $u \in U$ that has a neighbour $v \in V_1$.
Then $G[v_2,v_1,v,u,v_3,v_4,v_5]$ is an $S_{1,2,3}$, a contradiction.
We conclude that~$U$ must be empty.

Now since~$G$ is $K_3$-free, for every~$i$ the set~$V_i$ is anti-complete to the set~$V_{i+2}$.
Moreover, if~$i$ and~$j$ are such that the vertices~$v_i$ and~$v_j$ are at distance more than~$2$ on the cycle, then~$V_i$ and~$V_j$ must be anti-complete, as otherwise there would be a smaller odd cycle than~$C$ in~$G$, which would contradict the minimality of~$k$.

Note that every set~$V_i$ is independent in~$G$, since~$G$ is $K_3$-free.
If a vertex $x_1 \in\nobreak V_1$ is non-adjacent to a vertex $x_2 \in V_2$ then $G[v_3,x_2,v_2,x_1,v_4,v_5,v_6]$ is an~$S_{1,2,3}$, a contradiction.
Therefore a vertex $x_i \in V_i$ is adjacent to a vertex $x_j \in V_j$ if and only if~$v_i$ and~$v_j$ are consecutive vertices of~$C$.
In other words, for every~$i$, every vertex in~$V_i$ is a false twin of~$v_i$.
By Lemma~\ref{lem:false-twin} we may therefore assume that every~$V_i$ is empty, so~$G$ is an induced odd cycle.
By Lemma~\ref{lem:atmost-2}, $G$ has clique-width at most~$4$.
\qedllncs
\end{proof}

In our second lemma we state a number of sufficient conditions for a subclass of $(K_3,S_{1,2,3})$-free graphs to be of bounded clique-width when~$C_5$ is no longer a forbidden induced subgraph. To prove it we will need Lemmas~\ref{lem:123-imp} and~\ref{lem:k-decomp}.

\begin{lemma}\label{lem:gen}
Let~${\cal G}$ be the subclass of $(K_3,S_{1,2,3})$-free graphs for which the vertices in each graph $G \in {\cal G}$ can be partitioned into ten independent sets $V_1,\ldots,V_5,\allowbreak W_1,\ldots,W_5$, such that the following seven conditions hold (we interpret subscripts modulo~$5$):
\begin{enumeratei}
\item for all~$i$, $V_i$ is anti-complete to $V_{i-2} \cup V_{i+2} \cup W_{i-1} \cup W_{i+1};$\label{ass:vi-anti-vi2wi1}

\item for all~$i$, $W_i$ is complete to $W_{i-1}\cup W_{i+1}$;\label{ass:wi-comp-wi1}

\item for all~$i$, every vertex of~$V_i$ is trivial to at least one of the sets~$V_{i+1}$ and $V_{i-1}$;\label{ass:vi-triv-v-1-or-v+1}

\item for all~$i$, every vertex in~$V_i$ is trivial to~$W_i$;\label{ass:vi-triv-wi}

\item for all~$i$, $W_i$ is trivial to~$W_{i-2}$ and to~$W_{i+2}$;\label{ass:wi-triv-wi2}

\item for all $i, j$, the graphs induced by $V_i \cup V_j$ and $V_i \cup W_j$ are $P_7$-free;\label{ass:p7-free}

\item for all~$i$, there are no three vertices~$v \in V_i$, $w \in V_{i+1}$ and $x \in W_{i+3}$ such that $v, w$ and~$x$ are pairwise non-adjacent.\label{ass:no-3P_1}
\end{enumeratei}
Then~${\cal G}$ has bounded clique-width.
\end{lemma}

\begin{proof}
\begin{sloppypar}
Let~$G$ be a $(K_3,S_{1,2,3})$-free graph with such a partition that satisfies Conditions~\ref{ass:vi-anti-vi2wi1}--\ref{ass:no-3P_1} of the lemma.
Note that for all~$i$, every vertex $v \in V_i$ is trivial to $V_{i-2},V_{i+2},W_{i-1},W_{i+1},W_i$ and either trivial to~$V_{i-1}$ or trivial to~$V_{i+1}$.
Therefore a vertex $v \in V_i$ can only be non-trivial to $W_{i-2}, W_{i+2}$ and at most one of~$V_{i-1}$ and~$V_{i+1}$.
Likewise, every vertex $w \in W_i$ is trivial to $W_{i-1}, W_{i+1}, W_{i-2}, W_{i+2}, V_{i-1}$ and~$V_{i+1}$.
Therefore, a vertex $w \in W_i$ can only be non-trivial to~$V_i, V_{i-2}$ and~$V_{i+2}$ (and every vertex in~$V_i$ is trivial to~$W_i$).
\end{sloppypar}

For $i \in \{1,\ldots,5\}$, let~$W_i'$ be the set of vertices in~$W_i$ that are non-trivial to both~$V_{i-2}$ and~$V_{i+2}$, let~$V_i'$ be the set of vertices in~$V_i$ that are non-trivial to both~$V_{i+1}$ and~$W_{i-2}$ and let~$V_i''$ be the set of vertices in~$V_i$ that are non-trivial to both~$V_{i-1}$ and~$W_{i+2}$.
Note that $V_i' \cap V_i'' = \emptyset$ by Condition~\ref{ass:vi-triv-v-1-or-v+1}.

We say that an edge is {\em irrelevant} if one of its end-vertices is in a set $V_i, V_i', V_i'', W_i$ or~$W_i'$, and its other end-vertex is complete to this set, otherwise we say that the edge is {\em relevant}.
We will now show that for $i \in \{1,\ldots,5\}$, the graph $G[V_i' \cup V_{i+1}'' \cup W_{i-2}']$ can be separated from the rest of~$G$ by using a bounded number of bipartite complementations.
To do this, we first prove the following claim.

\clm{\label{clm:split-off} If $u \in V_i' \cup V_{i+1}'' \cup W_{i-2}'$ and $v \notin V_i' \cup V_{i+1}'' \cup W_{i-2}'$ are adjacent then~$uv$ is an irrelevant edge.}
We split the proof of Claim~\ref{clm:split-off} into the following cases.

\thmcase{\label{case:u-in-V_i} $u \in V_i'$.}
Since~$u$ is in~$V_i$, $v$ must be in $V_{i-1} \cup V_{i+1} \cup W_{i-2} \cup W_{i+2}$, otherwise~$uv$ would be irrelevant by Condition~\ref{ass:vi-anti-vi2wi1} or~\ref{ass:vi-triv-wi}.
We consider the possible cases for~$v$.

\thmsubcase{$v \in V_{i-1}$.}
Since~$u$ is in~$V_i'$, it is non-trivial to $V_{i+1}$, so by Condition~\ref{ass:vi-triv-v-1-or-v+1}, $u$ is trivial to~$V_{i-1}$. Therefore~$uv$ is irrelevant.

\thmsubcase{\label{case:special-3P1-K3-case} $v \in V_{i+1}$.}
Suppose, for contradiction, that~$v$ is complete to~$W_{i-2}$.
Let $w \in W_{i-2}$ be a neighbour of~$u$ (such a vertex~$w$ exists, since~$u$ is non-trivial to~$W_{i-2}$).
Then $G[u,v,w]$ is a~$K_3$, a contradiction, so~$v$ cannot be complete to~$W_{i-2}$.
Now suppose, for contradiction that~$v$ is anti-complete to~$W_{i-2}$.
We may assume that~$v$ has a non-neighbour $u' \in V_i'$, otherwise~$v$ would be trivial to~$V_i'$, in which case~$uv$ would be irrelevant.
Since $u' \in V_i'$, $u'$ is non-trivial to~$W_{i-2}$, so it must have a non-neighbour $w \in W_{i-2}$.
Then, since~$v$ is anti-complete to~$W_{i-2}$, it follows that $G[u',v,w]$ is a~$3P_1$, contradicting Condition~\ref{ass:no-3P_1}.
We may therefore assume that~$v$ is non-trivial to~$W_{i-2}$.
We know that $v \notin V_{i+1}''$.
Therefore~$v$ must be trivial to~$V_i$, so~$uv$ is irrelevant.

\thmsubcase{$v \in W_{i-2}$.}
Reasoning as in the previous case, we find that~$v$ cannot be complete or anti-complete to~$V_{i+1}$.
Hence, as $v \notin W_{i-2}'$, $v$ must  be trivial to~$V_i$, so~$uv$ is irrelevant.

\thmsubcase{$v \in W_{i+2}$.}
Since~$u$ is non-trivial to~$W_{i-2}$ (by definition of~$V_i'$), there is a vertex $w \in W_{i-2}$ that is adjacent to~$u$.
By Condition~\ref{ass:wi-comp-wi1}, $w$ is adjacent to~$v$.
Therefore $G[u,v,w]$ is a~$K_3$.
This contradiction implies that $v \notin W_{i+2}$.
This completes Case~\ref{case:u-in-V_i}.

\medskip
\noindent
Now assume that $u \notin V_i'$.
Then, by symmetry, $u \notin V_{i+1}''$.
This means that the following case holds.

\thmcase{\label{case:u-in-Wi2} $u \in W_{i-2}'$.}
We argue similarly to Case~\ref{case:special-3P1-K3-case}.
We may assume that~$v$ is non-trivial to~$W_{i-2}'$, otherwise~$uv$ would be irrelevant.
By Conditions~\ref{ass:vi-anti-vi2wi1},~\ref{ass:wi-comp-wi1}, \ref{ass:vi-triv-wi} and~\ref{ass:wi-triv-wi2}, it follows that $v \in V_i \cup V_{i+1}$.
Without loss of generality assume that $v \in V_i$.
Since $v \notin V_i'$ and~$v$ is non-trivial to~$W_{i-2}$, it follows that~$v$ is trivial to~$V_{i+1}$.
If~$v$ is complete to~$V_{i+1}$ then since~$u$ is non-trivial to~$V_{i+1}$, there must be a vertex~$w \in V_{i+1}$ adjacent to~$u$, in which case $G[u,v,w]$ is a~$K_3$, a contradiction.
Therefore~$v$ must be anti-complete to~$V_{i+1}$.
Since~$v$ is non-trivial to~$W_{i-2}'$, there must be a vertex $u' \in W_{i-2}'$ that is non-adjacent to~$v$.
Since $u' \in W_{i-2}'$, $u'$ must have a non-neighbour $w \in V_{i+1}$.
Then $G[u',v,w]$ is a~$3P_1$, contradicting Condition~\ref{ass:no-3P_1}.
This completes Case~\ref{case:u-in-Wi2}.

\medskip
\noindent
We conclude that, if $u \in V_i' \cup V_{i+1}'' \cup W_{i-2}'$ and $v \notin V_i' \cup V_{i+1}'' \cup W_{i-2}'$ are adjacent, then~$uv$ is an irrelevant edge. 
Hence we have proven Claim~\ref{clm:split-off}.

\medskip
\noindent
By Claim~\ref{clm:split-off} we find that if $u \in V_i' \cup V_{i+1}'' \cup W_{i-2}'$ and $v \notin V_i' \cup V_{i+1}'' \cup W_{i-2}'$ are adjacent then~$u$ or~$v$ is complete to some set $V_j, V_j', V_j'', W_j$ or~$W_j'$ that contains~$v$ or~$u$, respectively.
By applying a bounded number of bipartite complements (which we may do by Fact~\ref{fact:bip}), we can separate $G[V_i' \cup V_{i+1}'' \cup W_{i-2}']$ from the rest of~$G$.
By Conditions~\ref{ass:p7-free} and~\ref{ass:no-3P_1}  and the fact that~$G$ is $(K_3,S_{1,2,3})$-free, Lemmas~\ref{lem:123-imp} and~\ref{lem:k-decomp} imply that $G[V_i' \cup V_{i+1}'' \cup W_{i-2}']$ has clique-width at most~$6$.
Repeating this argument for each~$i$, we may assume that $V_i' \cup V_{i+1}'' \cup W_{i-2}' = \emptyset$ for every~$i$.

\medskip
For $i \in \{1,\ldots,5\}$ let~$V_i^*$ be the set of vertices in~$V_i$ that are either non-trivial to~$V_{i+1}$ or non-trivial to~$W_{i+2}$ and let~$V_i^{**}$ be the set of the remaining vertices in~$V_i$.
For $i \in \{1,\ldots,5\}$, let~$W_i^*$ be the set of vertices that are non-trivial to~$V_{i+2}$ and let~$W_i^{**}$ be the set of the remaining vertices in~$W_i$.

We claim that every vertex in~$V_i$ that is non-trivial to~$V_{i-1}$ or that is non-trivial to~$W_{i-2}$ is in~$V_i^{**}$.
Indeed, if $v \in V_i$ is non-trivial to~$V_{i-1}$ then by Condition~\ref{ass:vi-triv-v-1-or-v+1}, $v$ is trivial to~$V_{i+1}$ and since~$V_i''$ is empty, $v$ must be trivial to~$W_{i+2}$.
If $v \in V_i$ is non-trivial to~$W_{i-2}$ then~$v$ must be trivial to~$V_{i+1}$ since~$V_i'$ is empty.
Moreover, in this case~$v$ must also be trivial to~$W_{i+2}$, otherwise, by Condition~\ref{ass:wi-comp-wi1} the vertex~$v$, together with a neighbour of~$v$ in each of~$W_{i+2}$ and~$W_{i-2}$, would induce a~$K_3$ in~$G$.
It follows that every vertex in~$V_i$ that is non-trivial to~$V_{i-1}$ or that is non-trivial to~$W_{i-2}$ is indeed in~$V_i^{**}$.
Similarly, for all~$i$, since~$W_i'$ is empty, every vertex in~$W_i$ that is non-trivial to~$V_{i-2}$ is in~$W_i^{**}$.

We say that an edge~$uv$ is {\em insignificant} if~$u$ or~$v$ is in some set $V_i^*,V_i^{**},W_i^*$ or~$W_i^{**}$ and the other vertex is trivial to this set;
all other edges are  said to be {\em significant}.
We prove the following claim.

\clm{\label{clm:split-off2} If $u \in W_i^* \cup V_{i+2}^{**} \cup V_{i+1}^* \cup W_{i-2}^{**}$ and $v \notin W_i^* \cup V_{i+2}^{**} \cup V_{i+1}^* \cup W_{i-2}^{**}$ are adjacent then the edge~$uv$ is insignificant.}
To prove this claim suppose, for contradiction, that~$uv$ is a significant edge. We split the proof into two cases.

\thmcase{$u \in W_i$.}
We will show that $v \in V_{i+2}^{**}$ or $v \in V_{i-2}^*$ if $u \in W_i^*$ or $u \in W_i^{**}$, respectively.
By Conditions~\ref{ass:vi-anti-vi2wi1}, \ref{ass:wi-comp-wi1}, \ref{ass:vi-triv-wi} and~\ref{ass:wi-triv-wi2} we know that~$u$ is trivial to $V_{i-1}$, $V_{i+1}$, $W_{i-1}$, $W_{i+1}$, $W_{i-2}$ and~$W_{i+2}$, and that every vertex of~$V_i$ is trivial to~$W_i$.
Furthermore, $u$ is trivial to~$W_i^{**}\setminus\{u\}$ since~$W_i$ is independent.
Therefore $v \in V_{i-2} \cup V_{i+2}$.
Note that~$v$ is non-trivial to~$W_i$ (by choice of~$v$).
If $u \in W_i^*$ then~$u$ must be trivial to~$V_{i-2}$, since~$W_i'$ is empty.
Therefore $v \in V_{i+2}$.
Now if $v \in V_{i+2}^*$ then~$v$ is non-trivial to~$V_{i-2}$ or non-trivial to~$W_{i-1}$.
In the first case~$v$ is non-trivial to both~$V_{i-2}$ and~$W_i$, contradicting the fact that~$V_{i+2}'$ is empty.
In the second case~$v$ has a neighbour $w \in W_{i-1}$.
By Condition~\ref{ass:wi-comp-wi1}, $w$ is adjacent to~$u$, so $G[u,v,w]$ is a~$K_3$.
This contradiction implies that if $u \in W_i^*$ then $v \in V_{i+2}^{**}$, contradicting the choice of~$v$.
Now suppose $u \in W_i^{**}$.
Then~$u$ is trivial to~$V_{i+2}$, so $v \in V_{i-2}$.
If $v \in V_{i-2}^{**}$ then $v$ is trivial~$W_i$ (by definition of $V_{i-2}^{**}$).
Therefore if $u \in W_i^{**}$ then $v \in V_{i-2}^*$, contradicting the choice of~$v$.

\medskip
\noindent
We conclude that for every $i \in \{1,\ldots,5\}$ the vertex~$u$ is not in~$W_i$. Similarly, we may assume~$v \notin W_i$.
This means that the following case holds.

\thmcase{$u \in V_i$, $v \in V_j$ for some $i,j$.}
Then $i \neq j$, since $V_i$ is an independent set.
By Condition~\ref{ass:vi-anti-vi2wi1}, $j \notin \{i-2,i+2\}$.
Without loss of generality, we may therefore assume that $j=i+1$.
If $u \in V_i^{**}$ then~$u$ is trivial to~$V_{i+1}$, so we may assume that $u \in V_i^*$.
If $v \in V_{i+1}^*$ then~$v$ is non-trivial to~$V_{i+2}$, so by Condition~\ref{ass:vi-triv-v-1-or-v+1} $v$ is trivial to~$V_i$, contradicting the fact that~$uv$ is significant.
Therefore $v \in V_{i+1}^{**}$, contradicting the choice of~$v$.

\medskip
\noindent
We conclude that if for some~$i$, $u \in W_i^* \cup V_{i+2}^{**} \cup V_{i+1}^* \cup W_{i-2}^{**}$ and $v \notin W_i^* \cup V_{i+2}^{**} \cup V_{i+1}^* \cup W_{i-2}^{**}$ are adjacent then the edge~$uv$ is insignificant.
Hence we have proven Claim~\ref{clm:split-off2}.

\medskip
\noindent
Note that $W_i^*,V_{i+2}^{**}, V_{i+1}^*$ and~$W_{i-2}^{**}$ are independent sets.
By Condition~\ref{ass:vi-anti-vi2wi1}, $W_i^*$ is anti-complete to~$V_{i+1}^*$ and~$V_{i+2}^{**}$ is anti-complete to~$W_{i-2}^{**}$.
Therefore $W_i^* \cup V_{i+1}^*$ and $V_{i+2}^{**} \cup W_{i-2}^{**}$ are independent sets.
Thus $G[W_i^* \cup V_{i+2}^{**} \cup V_{i+1}^* \cup W_{i-2}^{**}]$ is an $S_{1,2,3}$-free bipartite graph, which has bounded clique-width by Lemma~\ref{lem:bipartite}.
Applying a bounded number of bipartite complementations (which we may do by Fact~\ref{fact:bip}), we can separate $G[W_i^* \cup V_{i+2}^{**} \cup V_{i+1}^* \cup W_{i-2}^{**}]$ from the rest of the graph.
We may thus assume that $W_i^* \cup V_{i+2}^{**} \cup V_{i+1}^* \cup W_{i-2}^{**}= \emptyset$.
Repeating this process for each~$i$ we obtain the empty graph.
This completes the proof.
\qedllncs
\end{proof}

\section{The Four Triangle-free Cases}\label{s-triangle}
We can now give the following result, which also implies the $(K_3,\allowbreak P_1+\nobreak 2P_2)$-free case.

\begin{theorem}\label{thm:gen2}
For $H \in \{P_1+\nobreak P_5,\allowbreak S_{1,2,2},\allowbreak P_1+\nobreak P_2+\nobreak P_3\}$, the class of $(K_3,H)$-free graphs has bounded clique-width. 
\end{theorem}

The proofs for all three cases are broadly similar.
We will prove the $H=P_1+\nobreak P_2+\nobreak P_3$ case separately, as it is a little more involved than the other two cases.

\subsection{Proof of the $H = P_1+\nobreak P_5$ and $H=S_{1,2,2}$ Cases.}

\begin{proof}
Let $H \in \{P_1+\nobreak P_5,\allowbreak S_{1,2,2}\}$ and consider a $(K_3,H)$-free graph~$G$. We may assume that~$G$ is connected.

By Lemma~\ref{lem:noC5}, we may assume that~$G$ contains an induced cycle on five vertices, say $C=v_1-v_2-\cdots-v_5-v_1$.
Again, we will interpret subscripts on vertices and vertex sets modulo~$5$.

Since~$G$ is $K_3$-free, no vertex~$v$ is adjacent to two consecutive vertices of the cycle.
Therefore every vertex of~$G$ has either zero, one or two neighbours on the cycle and if it has two neighbours then they must be non-consecutive vertices of the cycle.

We partition the vertices of~$G$ that are not on~$C$ as follows:
\begin{itemize}
\item $U$: the set of vertices adjacent to no vertices of~$C$,
\item $W_i$: the set of vertices whose unique neighbour in~$C$ is~$v_i$ and
\item $V_i$: the set of vertices adjacent to~$v_{i-1}$ and~$v_{i+1}$.
\end{itemize}
In the remainder of the proof we will show how to modify the graph using operations that preserve boundedness of clique-width, such that in the resulting graph the set~$U$ is empty and 
the partition $V_1,\ldots,V_5,W_1,\ldots,W_5$ satisfies
Conditions~\ref{ass:vi-anti-vi2wi1}--\ref{ass:no-3P_1} 
of Lemma~\ref{lem:gen}. 
In order to do this we prove a number of claims.

The first two claims follow immediately from the fact that~$G$ is $K_3$-free.

\clm{\label{clm:1-vi-wi-indep} For all~$i$, $V_i$ and~$W_i$ are independent sets.}
\smallclm{\label{clm:1-vi-non-nbrs} For all~$i$, $V_i$ is anti-complete to $V_{i-2} \cup V_{i+2} \cup W_{i-1} \cup W_{i+1}$.}\\
\smallclm{\label{clm:1-u-empty} We may assume that~$U$ is empty.}
We prove Claim~\ref{clm:1-u-empty} as follows.
First consider the case where $H = S_{1,2,2}$ and suppose, for contradiction, that~$U$ is not empty.
Since~$G$ is connected there must be a vertex $u \in U$ that is adjacent to a vertex $v \notin U$ that has a neighbour on the cycle~$C$.
Without loss of generality, we may assume that~$v \in V_1 \cup W_2$, in which case~$v$ is adjacent to~$v_2$ and non-adjacent to $v_1,v_3$ and~$v_4$.
Now $G[v_2,v_1,v_3,v_4,v,u]$ is an~$S_{1,2,2}$.
This contradiction means that $U = \emptyset$ if $H = S_{1,2,2}$.

Now consider the case where $H=P_1+\nobreak P_5$ and suppose that~$U$ is non-empty.
Suppose, for contradiction, that there are two vertices $u,u' \in U$ that do not have the same neighbourhood in some set~$V_i$ or~$W_i$.
Without loss of generality, assume $v \in V_1 \cup W_2$ is adjacent to~$u$, but not~$u'$.
Note that~$v$ is adjacent to~$v_2$, but non-adjacent to $v_1,v_3$ and~$v_4$.
Then $G[v_4,u',u,v,v_2,v_1]$ is a $P_1+\nobreak P_5$ if~$u$ and~$u'$ are adjacent and
$G[u',u,v,v_2,v_3,v_4]$ is a $P_1+\nobreak P_5$ if they are not.
This contradiction means that every vertex in~$U$ has the same neighbourhood in every set~$V_i$ and every set~$W_i$.
Since~$G$ is connected there must be a vertex~$v$ in some~$V_i$ or~$W_i$ that is adjacent to every vertex of~$U$.
Since~$G$ is $K_3$-free, $U$ must therefore be an independent set.
Applying a bipartite complementation (which we may do by Fact~\ref{fact:bip}) between~$U$ and the vertices adjacent to the vertices of~$U$ disconnects~$U$ from the rest of the graph.
Since~$G[U]$ is independent, it has clique-width at most~$1$.
We may therefore assume that~$U$ is empty.

\clm{\label{clm:1-Wi-comp-Wi1} For all~$i$, $W_i$ is complete to $W_{i-1} \cup W_{i+1}$.}
Suppose, for contradiction, that $v \in W_1$ has a non-neighbour $w \in W_2$.
Then $G[w,\allowbreak v,\allowbreak v_1,\allowbreak v_5,\allowbreak v_4,v_3]$ is a $P_1+\nobreak P_5$ and $G[v_1,v,v_2,w,v_5,v_4]$ is an~$S_{1,2,2}$.
This contradiction proves the claim.

\smallskip
See \figurename~\ref{fig:K_3,2P_2+P_1} for an illustration of the graph $G$.

\begin{figure}[h]
\begin{center}
\begin{tikzpicture}
\coordinate (v1) at (90:2) ;
\coordinate (v2) at (90-72:2) ;
\coordinate (v3) at (90-144:2) ;
\coordinate (v4) at (90-216:2) ;
\coordinate (v5) at (90-288:2) ;
\coordinate (V1) at (90:1) ;
\coordinate (V2) at (90-72:1) ;
\coordinate (V3) at (90-144:1) ;
\coordinate (V4) at (90-216:1) ;
\coordinate (V5) at (90-288:1) ;
\coordinate (W1) at (90:3) ;
\coordinate (W2) at (90-72:3) ;
\coordinate (W3) at (90-144:3) ;
\coordinate (W4) at (90-216:3) ;
\coordinate (W5) at (90-288:3) ;

\tikzset{EdgeStyle/.style   = {line width     = 0.8pt, black}}

\draw (v1) -- (v2);
\draw (v2) -- (v3);
\draw (v3) -- (v4);
\draw (v4) -- (v5);
\draw (v1) -- (v5);

\draw (v1) -- (W1);
\draw (v2) -- (W2);
\draw (v3) -- (W3);
\draw (v4) -- (W4);
\draw (v5) -- (W5);
%

\draw (W1) -- (W2);
\draw (W2) -- (W3);
\draw (W3) -- (W4);
\draw (W4) -- (W5);
\draw (W1) -- (W5);

\draw (v1) -- (V2);
\draw (v2) -- (V3);
\draw (v3) -- (V4);
\draw (v4) -- (V5);
\draw (v5) -- (V1);
\draw (v1) -- (V5);
\draw (v2) -- (V1);
\draw (v3) -- (V2);
\draw (v4) -- (V3);
\draw (v5) -- (V4);

\draw [fill=black] (v1) circle (1.5pt) ;
\draw [fill=black] (v2) circle (1.5pt) ;
\draw [fill=black] (v3) circle (1.5pt) ;
\draw [fill=black] (v4) circle (1.5pt) ;
\draw [fill=black] (v5) circle (1.5pt) ;
\draw [fill=white] (V1) circle (3pt) ;
\draw [fill=white] (V2) circle (3pt) ;
\draw [fill=white] (V3) circle (3pt) ;
\draw [fill=white] (V4) circle (3pt) ;
\draw [fill=white] (V5) circle (3pt) ;
\draw [fill=white] (W1) circle (3pt) ;
\draw [fill=white] (W2) circle (3pt) ;
\draw [fill=white] (W3) circle (3pt) ;
\draw [fill=white] (W4) circle (3pt) ;
\draw [fill=white] (W5) circle (3pt) ;

\draw (W1) node [label={[label distance=-2pt]90:$W_1$}] {} ;
\draw (W2) node [label={[label distance=-4pt]90-72:$W_2$}] {} ;
\draw (W3) node [label={[label distance=-4pt]90-144:$W_3$}] {} ;
\draw (W4) node [label={[label distance=-4pt]90-216:$W_4$}] {} ;
\draw (W5) node [label={[label distance=-4pt]90-288:$W_5$}] {} ;

\draw (v1) node [label={[label distance=-4pt]54:$v_1$}] {} ;
\draw (v2) node [label={[label distance=-4pt]54-72:$v_2$}] {} ;
\draw (v3) node [label={[label distance=-2pt]54-144-3:$v_3$}] {} ;
\draw (v4) node [label={[label distance=-3pt]54-216-20:$v_4$}] {} ;
\draw (v5) node [label={[label distance=-2pt]54-288-34:$v_5$}] {} ;

\draw (V1) node [label={[label distance=-2pt]-90:$V_1$}] {} ;
\draw (V2) node [label={[label distance=-4pt]-90-72:$V_2$}] {} ;
\draw (V3) node [label={[label distance=-4pt]-90-144:$V_3$}] {} ;
\draw (V4) node [label={[label distance=-4pt]-90-216:$V_4$}] {} ;
\draw (V5) node [label={[label distance=-4pt]-90-288:$V_5$}] {} ;
\end{tikzpicture}

\caption{The graph~$G$. The black points are the vertices of the cycle~$C$. The circles are (possibly empty) independent sets of vertices and the lines are complete bipartite graphs. 
Note that~$G$ may contain additional edges that are not represented in this figure.\label{fig:K_3,2P_2+P_1}}
\end{center}
\end{figure}
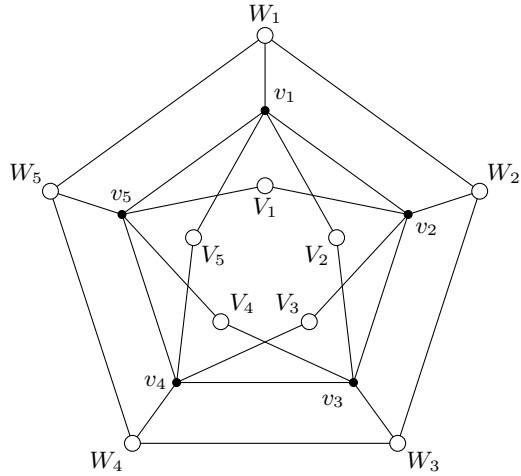

\clm{For all~$i$, every vertex of~$V_i$ is trivial to at least one of the sets~$V_{i+1}$ and~$V_{i-1}$.}
Suppose, for contradiction that the claim is false.
Without loss of generality, there is a vertex $v \in V_2$ with non-neighbours $u \in V_1$ and $w \in V_3$.
By Claim~\ref{clm:1-vi-non-nbrs}, $u$ and~$w$ must be non-adjacent.
Then $G[v_5,u,v_1,v,v_4,w]$ is an~$S_{1,2,2}$ and $G[u,v_1,v,v_3,v_4,w]$ is a $P_1+\nobreak P_5$. This contradiction completes the proof of the claim.

\clm{\label{clm:1-vi-triv-wi} For all~$i$, every vertex in~$V_i$ is trivial to~$W_i$.}
Suppose, for contradiction, that the claim is false.
Without loss of generality, we may assume there are vertices $v \in V_1$ and $w,w' \in W_1$ such that~$v$ is adjacent to~$w$, but not to~$w'$.
Then $G[v_2,v,v_1,w',v_3,v_4]$ is an~$S_{1,2,2}$ and $G[w',w,v,v_2,v_3,v_4]$ is a $P_1+\nobreak P_5$.
This contradiction completes the proof of the claim.

\clm{\label{clm:1-Wi-triv-Wi2} For all~$i$, $W_i$ is trivial to~$W_{i-2}$ and to~$W_{i+2}$.}
Suppose, for contradiction, that this does not hold.
Without loss of generality, assume $v \in W_1$ is adjacent to $w \in W_3$ and non-adjacent to $w' \in W_3$.
Then $G[v_1,v_2,v_5,v_4,v,w]$ is an~$S_{1,2,2}$ and $G[w',w,v,v_1,v_5,v_4]$ is a $P_1+\nobreak P_5$.
This contradiction proves the claim.

\clm{For all $i, j$, the graphs induced by $V_i \cup V_j$ and $V_i \cup W_j$ are $P_7$-free.}
Note that $P_1+\nobreak P_5$ is an induced subgraph of~$P_7$.
Therefore if $H=P_1+\nobreak P_5$ then the claim follows immediately.
Now suppose $H=S_{1,2,2}$.
Without loss of generality, we may assume $i=1$.
Suppose that $G[V_1 \cup V_j]$ or $G[V_1 \cup W_j]$ contains an induced~$P_7$, for some $i,j$.
By Claims \ref{clm:1-vi-wi-indep},~\ref{clm:1-vi-non-nbrs} and~\ref{clm:1-vi-triv-wi} and symmetry, we may assume that $G[V_1 \cup V_2]$ or $G[V_1 \cup W_3]$ contains this~$P_7$.
This~$P_7$ contains an induced subgraph isomorphic to~$2P_2$, say on vertices $v,v',w,w'$.
Then $G[v_5,v_4,v,v',w,w']$ is an~$S_{1,2,2}$. This contradiction completes the proof of the claim.

\clm{\label{clm:1-no-3P1} For all~$i$, there are no three vertices~$v \in V_i$, $w \in V_{i+1}$ and $x \in W_{i+3}$ such that $v, w$ and~$x$ are pairwise non-adjacent.}
Suppose, for contradiction that such pairwise non-adjacent vertices exist, say with $v \in\nobreak V_1,\allowbreak w \in V_2$ and $x \in W_4$.
Then $G[v_4,x,v_3,w,v_5,v]$ is an~$S_{1,2,2}$ and $G[x,v_3,w,v_1,v_5,v]$ is a $P_1+\nobreak P_5$.
This contradiction completes the proof of the claim.

\medskip
We now consider the graph obtained~$G'$ from~$G$ by removing the five vertices of~$C$.
Claims~\ref{clm:1-vi-wi-indep} and~\ref{clm:1-u-empty} show that we may assume $V_1,\ldots,V_5,W_1,\ldots,W_5$ are independent sets that form a partition of the vertex set of~$G'$.
Claims~\ref{clm:1-vi-non-nbrs} and~\ref{clm:1-Wi-comp-Wi1}--\ref{clm:1-no-3P1} correspond to the seven conditions of Lemma~\ref{lem:gen}.
Therefore~$G'$ has bounded clique-width.
By Fact~\ref{fact:del-vert}, $G$ also has bounded clique-width. This completes the proof.
\qedllncs
\end{proof}

\subsection{Proof of the $H=P_1+\nobreak P_2+\nobreak P_3$ Case.}

\begin{proof}
\setcounter{ctrclaim}{0}
Consider a $(K_3,P_1+\nobreak P_2+\nobreak P_3)$-free graph~$G$. We may assume that~$G$ is connected.

By Lemma~\ref{lem:noC5}, we may assume that~$G$ contains an induced cycle on five vertices, say $C=v_1-v_2-\cdots-v_5-v_1$.
Again, we will interpret subscripts on vertices and vertex sets modulo~$5$.

Since~$G$ is $K_3$-free, no vertex~$v$ is adjacent to two consecutive vertices of~$C$.
Therefore every vertex of~$G$ has either zero, one or two neighbours on~$C$ and if it has two neighbours then they must be non-consecutive vertices of~$C$.

We partition the vertices of~$G$ that are not on~$C$ as follows:
\begin{itemize}
\item $U$: the set of vertices adjacent to no vertices of~$C$,
\item $W_i$: the set of vertices whose unique neighbour in~$C$ is~$v_i$ and
\item $V_i$: the set of vertices adjacent to~$v_{i-1}$ and~$v_{i+1}$.
\end{itemize}
In the remainder of the proof we will show how to modify the graph using operations that preserve boundedness of clique-width, such that in the resulting graph the set~$U$ is empty and the partition $V_1,\ldots,V_5,W_1,\ldots,W_5$ satisfies Conditions~\ref{ass:vi-anti-vi2wi1}--\ref{ass:no-3P_1} of Lemma~\ref{lem:gen}.

The first two claims follow immediately from the fact that~$G$ is $K_3$-free.

\clm{\label{clm:2-vi-wi-indep} For all~$i$, $V_i$ and~$W_i$ are independent sets.}
\smallclm{\label{clm:2-vi-non-nbrs} For all~$i$, $V_i$ is anti-complete to $V_{i-2} \cup V_{i+2} \cup W_{i-1} \cup W_{i+1}$.}\\
\smallclm{\label{clm:2-u-empty} We may assume that~$U$ is empty.}
In order to proof Claim~\ref{clm:2-u-empty}, we first suppose that there are two adjacent vertices $u,u' \in U$.
Since~$G$ is connected, we may assume without loss of generality that~$u$ is adjacent to some vertex $v \in V_1 \cup W_2$.
Then~$u'$ must be non-adjacent to~$v$, otherwise $G[u,u',v]$ would be a~$K_3$.
Note that~$v$ is adjacent to~$v_2$, but not to $v_1,v_3$ or~$v_4$.
Now $G[v_1,v_3,v_4,u',u,v]$ is a $P_1+\nobreak P_2+\nobreak P_3$.
This contradiction implies that~$U$ must be an independent set. 

\begin{sloppypar}
Now suppose, for contradiction, that a vertex $u \in U$ has two neighbours in some set $V_i \cup W_{i+1}$.
Without loss of generality assume that~$u$ is adjacent to $v,v' \in V_1 \cup\nobreak W_2$.
Note that~$v$ and~$v'$ are adjacent to~$v_2$, but not adjacent to $v_1,v_3$ and~$v_4$.
Now $G[v_1,v_3,v_4,v,u,v']$ is a $P_1+\nobreak P_2+\nobreak P_3$.
This contradiction implies that every vertex of~$U$ has at most one neighbour in $V_i \cup W_{i+1}$ for each~$i$.
In particular, this means that every vertex of~$U$ has degree at most~$5$.
Therefore, if $u \in U$ then we delete $\{u\}\cup N(u)$ (a set of at most~$6$ vertices). This gives us a $(K_3,P_2+\nobreak P_3)$-free graph, which has bounded clique-width by Lemma~\ref{lem:diamond-P_2+P_3}.
By Fact~\ref{fact:del-vert}, we may therefore assume that~$U$ is empty, 
that is, we have proven Claim~\ref{clm:2-u-empty}.
\end{sloppypar}

\medskip
\noindent
We say that a set~$V_i$ or~$W_i$ is {\em large} if it contains at least two vertices and {\em small} if it contains exactly one vertex.
If any set~$V_i$ is not large then by Fact~\ref{fact:del-vert} we may assume that it is empty.
(Later in the proof, we may delete vertices from some sets~$V_i$ or~$W_i$. In doing so, some sets that were previously large may become small.
If this happens, we will simply repeat the argument.
We will only do this a bounded number of times, so boundedness of clique-width will be preserved.)

\clm{\label{clm:2-Wi-comp-Wi1} For all~$i$, $W_i$ is complete to $W_{i-1} \cup W_{i+1}$.}
Suppose, for contradiction, that $v \in W_1$ has a non-neighbour $w \in W_2$.
Since~$W_2$ is non-empty, it must be large, so it must contain a vertex~$w'$ distinct from~$w$.
Then $G[w,v_3,v_4,v_1,v,w']$ is a $P_1+\nobreak P_2+\nobreak P_3$ if~$v$ and~$w'$ are adjacent and $G[v,v_4,v_5,w,v_2,w']$ is a $P_1+\nobreak P_2+\nobreak P_3$ if they are not.
This contradiction completes the proof of Claim~\ref{clm:2-Wi-comp-Wi1}.

\clm{\label{clm:2-vi-triv-vi+1-or-vi-1}For all~$i$, every vertex of~$V_i$ is trivial to at least one of the sets~$V_{i+1}$ and~$V_{i-1}$.}
Suppose, for contradiction that the claim is false.
Without loss of generality, there is a vertex $v \in V_2$ with non-neighbours $u \in V_1$ and $w \in V_3$ and neighbour $u' \in V_1$.
By Claim~\ref{clm:2-vi-non-nbrs}, $w$ and must be non-adjacent to both~$u$ and~$u'$.
Then $G[u,v_4,w,v_1,v,u']$ is a $P_1+\nobreak P_2+\nobreak P_3$.
This contradiction completes the proof of Claim~\ref{clm:2-vi-triv-vi+1-or-vi-1}.

\clm{\label{clm:2-vi-triv-wi}For all~$i$, every vertex in~$V_i$ is trivial to~$W_i$.}
In fact we will prove a stronger statement, namely that for all~$i$, $V_i$ is trivial to~$W_i$.
Suppose, for contradiction, that this is not the case.
Without loss of generality, assume that~$V_1$ is not trivial to~$W_1$.
First suppose that there are vertices $w \in W_1$ and $v,v' \in V_1$ such that~$w$ is adjacent to~$v$, but not to~$v'$.
Then $G[v',v_3,v_4,v_1,w,v]$ is a $P_1+\nobreak P_2+\nobreak P_3$.
Therefore every vertex in~$W_1$ must be trivial to~$V_1$.
Since we assumed that~$V_1$ is not trivial to~$W_1$, there must therefore be vertices $v \in V_1$ and $w,w' \in W_1$ such that~$v$ is adjacent to~$w$, but not to~$w'$.
Since~$V_1$ is non-empty, it must be large, so there must be another vertex $v' \in V_1$.
Since every vertex of~$W_1$ is trivial to~$V_1$, $v'$ must be adjacent to~$w$ and non-adjacent to~$w'$. 
Then $G[w',v_3,v_4,v,w,v']$ is a $P_1+\nobreak P_2+\nobreak P_3$.
This contradiction completes the proof of Claim~\ref{clm:2-vi-triv-wi}.

\clm{\label{clm:2-Wi-anti-Wi2} We may assume that for all~$i$, $W_i$ is anti-complete to~$W_{i-2}$ and to~$W_{i+2}$.}
We start by showing that the edges between~$W_i$ and~$W_{i+2}$ form a matching.
Indeed, suppose for contradiction that there is a vertex $v \in W_1$ with two neighbours $w,w' \in W_3$.
Then $G[v_2,v_4,v_5,w,v,w']$ is a $P_1+\nobreak P_2+\nobreak P_3$, a contradiction.
By symmetry, no vertex of~$W_3$ has two neighbours in~$W_1$.
We conclude that the edges between~$W_i$ and~$W_{i+2}$ form a matching.

Let~$W_1'$ be the set of vertices in~$W_1$ that have a neighbour in~$W_3$. 
Similarly, let~$W_3''$ be the set of vertices in~$W_3$ that have a neighbour in~$W_1$.
Note that $|W_1'| = |W_3''|$ since the edges between~$W_1'$ and~$W_3''$ form a perfect matching.
We will show that every vertex of~$G \setminus (W_1' \cup W_3'')$ is trivial to~$W_1'$ and~$W_3''$.
This follows immediately if $|W_1'| = |W_3''| = 1$.

Assume $|W_1'| = |W_3''| \geq 2$.
Suppose there is a vertex $w \in V(G) \setminus (W_1' \cup W_3'')$ that is non-trivial to~$W_1'$.
Then we may choose $u,u' \in W_1'$ and $v,v' \in W_3''$ such that~$u$ is adjacent to~$v$ and~$w$, but non-adjacent to~$v'$ while~$u'$ is adjacent to~$v'$, but non-adjacent to~$v$ and~$w$.
Since~$w$ is non-trivial to~$W_1$, it cannot be in~$W_1$ (by Claim~\ref{clm:2-vi-wi-indep}), $V_2 \cup V_5$ (by Claim~\ref{clm:2-vi-non-nbrs}), $W_2 \cup W_5$ (by Claim~\ref{clm:2-Wi-comp-Wi1}), $V_1$ (by Claim~\ref{clm:2-vi-triv-wi}) or~$W_3$ (since we assumed $w \notin W_3''$).
Furthermore, $w \notin C$ by definition of~$W_1$.
Therefore $w \in V_4 \cup W_4 \cup V_3$.
By Claims~\ref{clm:2-vi-non-nbrs}, \ref{clm:2-Wi-comp-Wi1} and~\ref{clm:2-vi-triv-wi} respectively, 
we conclude that $w$ is trivial to~$W_3$.
Since~$u$ is adjacent to~$v$ and~$w$, it follows that~$w$ must be non-adjacent to~$v$, otherwise $G[u,v,w]$ would be a~$K_3$, a contradiction.
Therefore~$w$ must be anti-complete to~$W_3$.
If $w \in V_3 \cup W_4$, let $z=v_5$ and otherwise (if $w \in V_4$) let $z=v_4$.
Then~$z$ is non-adjacent to $u,u',v,v'$ and~$w$.
Now $G[z,u',v',v,u,w]$ is a $P_1+\nobreak P_2+\nobreak P_3$, a contradiction.
Therefore every vertex in $V(G) \setminus (W_1' \cup W_3'')$ is trivial to~$W_1'$.
By symmetry, 
every vertex in 
$V(G) \setminus (W_1' \cup W_3'')$ is trivial to~$W_3''$.

Therefore, by applying a bipartite complementation (which we may do by Fact~\ref{fact:bip}) between~$W_1'$ and the vertices in $V(G) \setminus W_3''$ that are complete to~$W_1'$ and another bipartite complementation between~$W_3''$
and the vertices in $V(G) \setminus W_1'$ that are complete to~$W_3''$, we separate $G[W_1' \cup W_3'']$ from the rest of the graph.
Since $G[W_1' \cup W_3'']$ is a perfect matching, it has clique-width at most~$2$.
We may therefore assume that $W_1' \cup W_3''$ is empty i.e. that~$W_1$ is anti-complete to~$W_3$.
Repeating this argument for each~$i \in \{1,\ldots,5\}$, we show that we may assume that~$W_i$ is anti-complete to~$W_{i-2}$ for every~$i$.
This completes the proof of Claim~\ref{clm:2-Wi-anti-Wi2}.

\medskip
\noindent
Note that when applying Claim~\ref{clm:2-Wi-anti-Wi2} we may delete vertices in some sets~$W_i$, which may cause some large sets to become small.
In this case, as stated earlier, we may simply delete the small sets as before.
Thus we may assume that every set~$W_i$ is either large or empty.

\clm{\label{clm:2-p7-free}For all $i, j$, the graphs induced by $V_i \cup V_j$ and $V_i \cup W_j$ are $P_7$-free.}
Suppose, for contradiction, that the claim is false.
Then there is an~$i$ and a~$j$ such that $G[V_i \cup V_j]$ or $G[V_i \cup W_j]$ contains an induced~$P_7$, say on vertices $u_1,\ldots,u_7$.
There must be a vertex $v_k \in C$ that is non-adjacent to every vertex of $V_i \cup V_j$ or $V_i \cup W_j$, respectively (since every vertex not in~$C$ has at most two neighbours in~$C$).
Then $G[v_k,u_1,u_2,u_4,u_5,u_6]$ is a $P_1+\nobreak P_2+\nobreak P_3$, a contradiction.
This completes the proof of Claim~\ref{clm:2-p7-free}.

\clm{\label{clm:2-no-3P1} For all~$i$, if there are vertices~$v \in V_i$, $w \in V_{i+1}$ and $x \in W_{i+3}$ such that $v, w$ and~$x$ are pairwise non-adjacent then~$G$ has bounded clique-width.}
Suppose that such pairwise non-adjacent vertices exist, say with $v \in V_1, w \in V_2$ and $x \in W_4$.
We start by showing that $V_3 \cup V_4 \cup V_5 \cup W_1 \cup W_2 \cup W_3 \cup W_5$ is empty.

First suppose there is a vertex $y \in V_3$.
Then~$y$ is non-adjacent to~$v$ and~$x$ by Claim~\ref{clm:2-vi-non-nbrs}.
Then $G[x,v,v_5,v_3,w,y]$ or $G[v,v_1,w,x,v_4,y]$ is a $P_1+\nobreak P_2+\nobreak P_3$ if~$y$ is adjacent or non-adjacent to~$w$, respectively.
This contradiction implies that~$V_3$ is empty.
By symmetry~$V_5$ is also empty.

Next, suppose there is a vertex $y \in V_4$.
Then~$y$ is non-adjacent to~$v$ and~$w$ by Claim~\ref{clm:2-vi-non-nbrs}.
Then $G[v,v_1,w,v_4,x,y]$ or $G[y,v_4,x,v_2,v_1,w]$ is a $P_1+\nobreak P_2+\nobreak P_3$ if~$y$ is adjacent or non-adjacent to~$x$, respectively.
This contradiction implies that~$V_4$ is empty.

Next, suppose there is a vertex $y \in W_1$.
Then~$y$ is non-adjacent to~$w$ and~$x$ by Claims~\ref{clm:2-vi-non-nbrs} and~\ref{clm:2-Wi-anti-Wi2}, respectively.
Then $G[w,x,v_4,v_2,v,y]$ or $G[v,x,v_4,w,v_1,y]$ is a $P_1+\nobreak P_2+\nobreak P_3$ if~$y$ is adjacent or non-adjacent to~$v$, respectively.
This contradiction implies that~$W_1$ is empty.
By symmetry~$W_2$ is also empty.

Finally, suppose that~$W_3$ is not empty.
Then~$W_3$ must be large, so it contains two vertices, say~$y$ and~$y'$.
Then~$y$ and~$y'$ are each non-adjacent to~$w$ and adjacent to~$x$ by Claims~\ref{clm:2-vi-non-nbrs} and~\ref{clm:2-Wi-comp-Wi1}, respectively.
If~$y$ is non-adjacent to~$v$ then $G[v,v_1,w,v_4,x,y]$ would be a $P_1+\nobreak P_2+\nobreak P_3$, a contradiction.
Therefore~$y$ is adjacent to~$v$, and similarly~$y'$ is adjacent to~$v$.
Now $G[v_4,v_1,w,y,v,y']$ is a $P_1+\nobreak P_2+\nobreak P_3$.
This contradiction implies that~$W_3$ is empty.
By symmetry, we may assume that~$W_5$ is also empty.

The above means that $V_3 \cup V_4 \cup V_5 \cup W_1 \cup W_2 \cup W_3 \cup W_5$ is indeed empty, so $V(G)=V_1 \cup V_2 \cup W_4 \cup V(C)$.

Let~$V_1'$ and~$V_1''$ be the set of vertices in~$V_1$ that are anti-complete or complete to $\{w,x\}$, respectively.
Let~$V_2'$ and~$V_2''$ be the set of vertices in~$V_2$ that are anti-complete or complete to $\{v,x\}$, respectively.
Let~$W_4'$ and~$W_4''$ be the set of vertices in~$W_4$ that are anti-complete or complete to $\{v,w\}$, respectively.
Observe that $v \in V_1', w \in V_2'$ and $x \in W_4'$.
We will show that 
$V_1', V_1'', V_2', V_2'', W_4'$ and $W_4''$
form a partition of $V(G) \setminus V(C)$.

Suppose, for contradiction, that there is a vertex $v' \in V_1$ with exactly one neighbour in $\{w,x\}$.
Then $G[v,v_4,x,v',w,v_1]$ or $G[v,v_1,w,v_4,x,v']$ is a $P_1+\nobreak P_2+\nobreak P_3$ if this neighbour is~$w$ or~$x$, respectively.
Therefore every vertex of~$V_1$ is in $V_1' \cup V_1''$.
Similarly, every vertex of~$V_2$ is in $V_2' \cup V_2''$.

Suppose, for contradiction, that there is a vertex $x' \in W_4$ with exactly one neighbour in $\{v,w\}$.
Without loss of generality, suppose that~$x'$ is adjacent to~$v$, but not to~$w$.
Then $G[x,w,v_3,v_5,v,x']$ is a $P_1+\nobreak P_2+\nobreak P_3$.
Therefore every vertex of~$W_4$ is in $W_4' \cup W_4''$.
Thus every vertex of $V(G) \setminus V(C)$ is in $V_1' \cup V_1'' \cup V_2' \cup V_2'' \cup V_4' \cup V_4''$.

Observe that the remarks made above for $v,w$ and~$x$ also hold if one of these is replaced by a vertex of $V_1',V_2'$ or~$W_4'$, respectively.
Indeed, suppose $v' \in V_1' \setminus \{v\}$, then every vertex of~$W_4$ must be either complete or anti-complete to $\{v',w\}$.
Since the vertices of~$W_4'$ are non-adjacent to~$w$, but the vertices of~$W_4''$ are adjacent to~$w$, it follows  that~$W_4'$ is anti-complete to $\{v',w\}$ and that~$W_4''$ is complete to $\{v',w\}$.
Therefore~$W_4'$ is anti-complete to~$V_1'$, and~$W_4''$ is complete to~$V_1'$.
Since~$G$ is $K_3$-free and every vertex of $V_1'' \cup V_2''$ is adjacent to~$x$, it follows that~$V_1''$ is anti-complete to~$V_2''$.
Similarly, we conclude that $V_1',V_2'$ and~$W_4'$ are pairwise anti-complete, $V_1'',V_2''$ and~$W_4''$ are pairwise anti-complete and for every pair of sets $S \in \{V_1',V_2',W_4'\}$ and $T \in \{V_1'',V_2'',W_4''\}$ such that $(S,T) \notin \{(V_1',V_1''), (V_2',V_2''),(W_4',W_4'')\}$, $S$ and~$T$ are complete to each-other.

Now if we delete the vertices of~$C$ (which we may do by Fact~\ref{fact:del-vert}) and apply bipartite complementations
between $V_1'\& V_2''$, $V_1'\& W_4''$,
$V_2' \& V_1''$, $V_2' \& W_4''$, 
$W_4' \& V_1''$ and $W_4' \& V_2''$,
 we obtain an edgeless graph, which therefore has clique-width at most~$1$.
By Fact~\ref{fact:bip}, it follows that~$G$ has bounded clique-width.
This completes the proof of Claim~\ref{clm:2-no-3P1}.

\medskip
\noindent
We now consider the graph~$G'$ obtained from~$G$ by removing the five vertices of~$C$.
Claims~\ref{clm:2-vi-wi-indep} and~\ref{clm:2-u-empty} show that we may assume $V_1,\ldots,V_5,W_1,\ldots,W_5$ are independent sets that form a partition of the vertex set of~$G'$.
Claims~\ref{clm:2-vi-non-nbrs} and~\ref{clm:2-Wi-comp-Wi1}--\ref{clm:2-no-3P1} correspond to the seven conditions of Lemma~\ref{lem:gen}.
Therefore~$G'$ has bounded clique-width.
By Fact~\ref{fact:del-vert}, $G$ also has bounded clique-width. This completes the proof.
\qedllncs
\end{proof}

\section{The Diamond-free Case}\label{s-diamond}

In this section, we prove that $(\mbox{diamond},P_1+\nobreak 2P_2)$-free graphs have bounded clique-width. In order to do this, we first need to prove the following two lemmas.

\newpage
\begin{lemma}\label{lem:discon}
The class of disconnected $(\mbox{diamond},P_1+\nobreak 2P_2)$-free graphs has bounded clique-width.
\end{lemma}
\begin{proof}
If~$G$ is a disconnected $(\mbox{diamond},P_1+\nobreak 2P_2)$-free graph then it contains at least two components.
Therefore every component of~$G$ must be $(\mbox{diamond},2P_2)$-free and thus has bounded clique-width by Lemma~\ref{lem:diamond-P_2+P_3}.
We conclude that~$G$ has bounded clique-width.
\end{proof}

\begin{lemma}\label{lem:K4}
The class of $(\mbox{diamond},P_1+\nobreak 2P_2)$-free graphs that contain a~$K_4$ has bounded clique-width. 
\end{lemma}

\begin{proof}
Let~$G$ be a $(\mbox{diamond},P_1+\nobreak 2P_2)$-free graph containing an induced~$K_4$.
By Lemma~\ref{lem:discon}, we may assume that~$G$ is connected.
Let~$K$ be a maximum clique of~$G$ and note that $|K| \geq 4$.
We may assume that~$G$ contains vertices outside~$K$, otherwise~$G$ is a clique on at least four vertices, in which case it has clique-width~$2$.

Suppose there is a vertex~$v$ in~$G$ that is not in~$K$, but has at least two
neighbours $x,y \in K$.
By maximality of~$K$, there must be a vertex $z \in K$ that is not adjacent to~$v$.
However this means that $G[x,y,v,z]$ is a diamond, a contradiction.
Therefore every vertex not in~$K$ has at most one neighbour in~$K$.

Choose $v_1,v_2,v_3,v_4 \in K$ arbitrarily.
For $i \in \{1,2,3,4\}$, let~$V_i$ be the set of vertices not in~$K$ whose unique neighbour in~$K$ is~$v_i$.
Let~$U$ be the set of vertices not in~$K$ that do not have a neighbour in $\{v_1,v_2,v_3,v_4\}$.
Note that vertices of~$U$ may have neighbours in $K \setminus \{v_1,v_2,v_3,v_4\}$.

\clm{\label{clm:4-UViVj-P1+P2-free}For $i,j \in \{1,2,3,4\}$, $G[U\cup V_i \cup V_j]$ must be $(P_1+\nobreak P_2)$-free.}
Indeed, if~$G[U\cup V_1 \cup V_2]$ contains an induced $P_1+\nobreak P_2$ on vertices $y_1,y_2,y_3$, say, then $G[y_1,y_2,y_3,v_3,v_4]$ is a $P_1+\nobreak 2P_2$, a contradiction.
The claim follows by symmetry.

\clm{\label{clm:4-Vi-indep-or-edge}For $i \in \{1,2,3,4\}$, we may assume $G[V_i]$ is either a clique on at most two vertices or an independent set.}
If~$G[V_1]$ contains an induced~$P_3$ on vertices $y_1,y_2,y_3$, say, then $G[v_1,y_2,\allowbreak y_1,y_3]$ is a diamond, a contradiction.
Therefore~$G[V_1]$ is a disjoint union of cliques.
Claim~\ref{clm:4-UViVj-P1+P2-free} implies that~$G[V_1]$ is either a clique,
or else every clique in~$G[V_1]$ contains at most one vertex i.e.~$V_1$ is an independent set.

Suppose, for contradiction, that~$V_1$ is a clique on at least three vertices.
We will show that the clique-width of~$G$ is bounded in this case.
First suppose, for contradiction, that there is a vertex $u \in U \cup V_2 \cup V_3 \cup V_4$.
Since $G[\{u\} \cup V_1]$ is $(P_1+\nobreak P_2)$-free by Claim~\ref{clm:4-UViVj-P1+P2-free}, $u$ must be adjacent to all but at most one vertex of~$V_1$.
Let $x,y \in V_1$ be neighbours of~$u$.
Then $G[x,y,u,v_1]$ is a diamond, a contradiction.
We conclude that $U \cup V_2 \cup V_3 \cup V_4 = \emptyset$, so $V(G) = K \cup V_1$.
Deleting~$v_1$ we obtain a disconnected $(\mbox{diamond},P_1+\nobreak 2P_2)$-free graph, which has bounded clique-width by Lemma~\ref{lem:discon}.
Therefore~$G$ has bounded clique-width by Fact~\ref{fact:del-vert}.
Therefore if~$V_1$ is a clique then it contains at most two vertices.
The claim follows by symmetry.

\clm{\label{clm:4-Vi-indep-Vj-triv}For distinct $i,j \in \{1,2,3,4\}$, if~$V_i$ is an independent set then every vertex of~$V_j$ is either complete or anti-complete to~$V_i$.}
Indeed, this follows directly from Claim~\ref{clm:4-UViVj-P1+P2-free}, which states that $G[V_i \cup V_j]$ is $(P_1+\nobreak P_2)$-free. (Note that if~$V_j$ is a clique then it may contain a vertex that is complete to~$V_i$ and another that is anti-complete to~$V_i$.)
 
\clm{\label{clm:4-U-large}We may assume~$U$ contains at least three vertices.}
Suppose that~$U$ has at most two vertices.
By Fact~\ref{fact:del-vert} and Claim~\ref{clm:4-Vi-indep-or-edge}, we may remove every vertex of~$U$ and every vertex of~$V_i$ for those~$V_i$ that are cliques.
After this, by Claim~\ref{clm:4-Vi-indep-or-edge}, every set~$V_i$ will either be empty or an independent set.
Furthermore, for distinct $i,j \in \{1,2,3,4\}$, by Claim~\ref{clm:4-Vi-indep-Vj-triv}, every vertex of~$V_i$ is trivial to~$V_j$ and vice versa, so~$V_i$ is complete or anti-complete to~$V_j$.
By Fact~\ref{fact:bip}, we may apply a bipartite complementation between~$V_i$ and~$V_j$ if they are complete.
By Fact~\ref{fact:del-vert}, we may delete $v_1,v_2,v_3,v_4$.
We obtain a graph that is the disjoint union of a clique and at most four independent sets and therefore has clique-width at most~$2$. It follows that the graph~$G$ must also have had bounded clique-width.
We may therefore assume that~$U$ contains at least three vertices.
This completes the proof of the claim.

\medskip
We now consider a number of cases:

\thmcase{Every vertex of~$K$ has at most one neighbour outside of~$K$.}
By Fact~\ref{fact:comp}, we may remove all the edges connecting pairs of vertices in~$K$.
Let~$G'$ be the resulting graph and note that in~$G'$, every vertex of~$K$ has at most one neighbour.
Then $\cw(G') \leq \cw(G'\setminus K) +1$.
(Given a $k$-expression for $G'\setminus K$, whenever we create a vertex~$v$ that has a neighbour~$w$ in~$K$, we immediately create~$w$ with a special new label~$*$, take the disjoint union and join~$v$ to~$w$ by an edge. For any vertices in~$K$ with no neighbours outside of~$K$, we simply add them with label~$*$ at the end of the process. This will give a $(k+1)$-expression for~$G'$.)
Now $G' \setminus K=G \setminus K$.
Since~$V_1$ contains at most one vertex, by Fact~\ref{fact:del-vert}, it is sufficient to show that $G \setminus (V_1 \cup K)$ has bounded clique-width.
However, $G \setminus (V_1 \cup K)$ is $(\mbox{diamond},2P_2)$-free, since if it contained an induced~$2P_2$ then this, together with~$v_1$ would induce a $P_1+\nobreak 2P_2$ in~$G$.
Therefore $G \setminus (V_1 \cup K)$ has bounded clique-width by Lemma~\ref{lem:diamond-P_2+P_3} and therefore~$G$ also has bounded clique-width.
This completes the proof of this case.

\medskip
We may now assume that at least one vertex of~$K$ has at least two neighbours outside of~$K$.

\thmcase{Exactly one vertex of~$K$ has neighbours outside~$K$.}
Suppose that~$v_1$ is the only vertex of~$K$ that has neighbours outside of~$K$ (at least one vertex of~$K$ has a neighbour outside of~$K$ since~$G$ is connected and not a clique).
Now $G \setminus \{v_1\}$ is a disconnected $(\mbox{diamond},P_1+\nobreak 2P_2)$-free graph, so it has bounded clique-width by Lemma~\ref{lem:discon}.
By Fact~\ref{fact:del-vert}, $G$ also has bounded clique-width.
This completes the proof of this case.

\medskip
We may now assume that at least two vertices of~$K$ have neighbours outside of~$K$.
Without loss of generality, we may therefore assume that the following case holds.

\thmcase{$V_1$ contains at least two vertices and~$V_2$ contains at least one vertex.}
Fix $x,y,z \in V_1 \cup V_2$, with two of these vertices in~$V_1$ and one in~$V_2$.
If these vertices are pairwise adjacent then $G[x,y,v_1,z]$ would be a diamond, a contradiction.
We may therefore assume that~$x$ and~$y$ are non-adjacent.
Now every vertex of~$v \in U$ is either complete or anti-complete to $\{x,y\}$, otherwise $G[v,x,y]$ would be a $P_1 + \nobreak P_2$ in $G[U \cup V_1 \cup V_2]$, which would contradict Claim~\ref{clm:4-UViVj-P1+P2-free}.

Suppose~$u,v \in U$.
If~$u$ and~$v$ are adjacent then they cannot both be complete to $\{x,y\}$, otherwise $G[u,v,x,y]$ would be a diamond and they cannot both be anti-complete to $\{x,y\}$, otherwise
$G[x,u,v]$ would be a $P_1+\nobreak P_2$ in $G[U \cup V_1 \cup V_2]$, which would contradict Claim~\ref{clm:4-UViVj-P1+P2-free}.
Therefore if~$u$ and~$v$ are adjacent then one of them is complete to $\{x,y\}$ and the other is anti-complete to $\{x,y\}$.
If~$u$ and~$v$ are non-adjacent then they must either both be complete to $\{x,y\}$ or both be anti-complete to $\{x,y\}$.
Indeed, suppose for contradiction that~$u$ is complete to $\{x,y\}$ and~$v$ is anti-complete to $\{x,y\}$.
Then $G[v,u,x]$ would be an induced $P_1+\nobreak P_2$ in $G[U \cup V_1 \cup V_2]$, which would contradict Claim~\ref{clm:4-UViVj-P1+P2-free}.
The above holds for every pair of vertices $u,v \in U$. 
This implies that~$G[U]$ is a complete bipartite graph with one of the sets in the bipartition consisting of the vertices complete to $\{x,y\}$ and the other consisting of the vertices anti-complete to $\{x,y\}$.
(Note that one of the parts of the complete bipartite graph~$G[U]$ may be empty, as we allow the case where~$U$ is an independent set.)

Note that the arguments in the above paragraph only used the facts that $G[U \cup V_1 \cup V_2]$ is $(P_1+\nobreak P_2,\mbox{diamond})$-free and that $V_1 \cup V_2$ contains two non-adjacent vertices.
Let~$U_1$ and~$U_2$ be the independent sets that form the bipartition of~$U$.
Note that since~$U$ contains at least three vertices (by Claim~\ref{clm:4-U-large}), we may assume without loss of generality that~$U_1$ contains at least two vertices.
If~$U_2$ contains exactly one vertex, by Fact~\ref{fact:del-vert}, we may delete it.
(Note that this may cause~$U$ to contain only two vertices, rather than  at least three, however this does not affect our later arguments.)
We may therefore assume that~$U_2$ is either empty or contains at least two vertices.
Repeating the argument in the previous paragraph with the roles of~$U$ and $V_1 \cup V_2$ reversed, we find that $G[V_1 \cup V_2]$ is a complete bipartite graph, with one side of the bipartition complete to~$U_1$ and the other anti-complete to~$U_1$ and if~$U_2$ is non-empty then one side of the bipartition is complete to~$U_2$ and the other is anti-complete to~$U_2$.
Similarly, for each pair of distinct $i,j \in \{1,2,3,4\}$, the same argument shows that $G[V_i \cup V_j]$ is also a complete bipartite graph with a similar bipartition.

We now proceed as follows: if~$V_i$ is a clique for some~$i$ then it contains at most two vertices, in which case we delete them and make~$V_i$ empty.
For every pair of distinct $i,j \in \{1,2,3,4\}$ ($V_i$ or~$V_j$ may be empty) $G[V_i \cup V_j]$ must then be an independent set, in which case we do nothing, or a complete bipartite graph with bipartition $(V_i,V_j)$, in which case we apply a bipartite complementation between~$V_i$ and~$V_j$.
Now every set~$V_i$ is either complete or anti-complete to~$U_1$ and complete or anti-complete to~$U_2$.
Applying at most $4 \times 2 = 8$ bipartite complementations, we can remove all edges between $V_1\cup \cdots \cup V_4$ and~$U$.
Next, we apply a bipartite complementation between~$U_1$ and~$U_2$.
Finally, we apply a complementation to the clique~$K$.
Let~$G'$ be the resulting graph and note that $G'[V_1 \cup \cdots \cup V_4 \cup U]$ and $G'[K]$ are independent sets and that in~$G'$ every vertex in $V_1 \cup \cdots \cup V_4 \cup U$ has at most one neighbour in~$K$.
Therefore~$G'$ is a disjoint union of stars, and so has clique-width at most~$2$.
By Facts \ref{fact:del-vert}, \ref{fact:comp} and~\ref{fact:bip}, it follows that~$G$ also has bounded clique-width.
This completes proof for this case and therefore completes the proof of the lemma.
\qedllncs
\end{proof}

\bigskip
To prove the main result of this section, we will need an additional notion.
Let~$G$ be a graph.
For each set~$T$ that induces a triangle in~$G$, let~$U^T$ be the set of vertices in~$G$ that have no neighbour in~$T$.
Let ${\cal U} = \{u \in U^T \; | \; T \; \mbox{\rm induces a triangle in} \; G\}$.
We say that the graph~$G$ is {\em basic} if we can partition the vertices of $G\setminus {\cal U}$ into three sets $V_1,V_2,V_3$ and also into sets $T^1,W_1,T^2,W_2,\ldots,T^p,W_p$ for some~$p$ such that the following properties hold:
\begin{enumeratei}
\item \label{prop:3-no-U-triangle}No triangle in~$G$ contains a vertex of~${\cal U}$.
\item \label{prop:3-UT-nbhd} For every triangle~$T$, the set~$U^T$ is independent and there is a vertex~$x \in V(T)$ such that $N(x) = N(u) \cup (V(T) \setminus \{x\})$ for all $u \in U^T$.
\item \label{prop:3-Vi-indep}$V_1,V_2$ and~$V_3$ are independent.
\item \label{prop:3-have-all-K3s}$\{G[T^1],\ldots,G[T^p]\}$ is the set of all induced triangles in~$G$ and each of them has exactly one vertex in each of $V_1,V_2$ and~$V_3$.
\item \label{prop:3-GW_i-nice}$G[W_i]$ is $(P_1+\nobreak 2P_2)$-free and does not contain an induced~$3P_1$ with one vertex in each of $V_1,V_2$ and~$V_3$.
\item If $i<j$ and $k +1 \not\equiv \ell \; (\bmod\; 3)$ then:
\begin{enumerate1}
\item \label{prop:3-Ti-Tj-anti}$T^i \cap V_k$ is anti-complete to $T^j\cap V_\ell$,
\item \label{prop:3-Ti-Wj-anti}$T^i \cap V_k$ is anti-complete to $W_j\cap V_\ell$,
\item \label{prop:3-Wi-Tj-anti}$W_i \cap V_k$ is anti-complete to $T^j\cap V_\ell$ and
\item \label{prop:3-Wi-Wj-anti}$W_i \cap V_k$ is anti-complete to $W_j\cap V_\ell$,
\end{enumerate1}
\item If $i<j$ and $k +1 \equiv \ell \; (\bmod\; 3)$ then:
\begin{enumerate1}
\item \label{prop:3-Ti-Tj-comp}$T^i \cap V_k$ is complete to $T^j \cap V_\ell$,
\item \label{prop:3-Ti-Wj-comp}$T^i \cap V_k$ is complete to $W_j \cap V_\ell$ and
\item \label{prop:3-Wi-Tj-comp}$W_i \cap V_k$ is complete to $T^j \cap V_\ell$.
\end{enumerate1}
\item If $i+1<j$ and $k +1 \equiv \ell \; (\bmod\; 3)$ then:
\begin{enumerate1}
\item \label{prop:3-Wi-Wj-comp}$W_i \cap V_k$ is complete to $W_j \cap V_\ell$.
\end{enumerate1}
\item If $i+1=j$ and $k+1 \equiv \ell \; (\bmod\; 3)$ then:
\begin{enumerate1}
\item \label{prop:3-Wi-Wi+1-triv}$W_i \cap V_k$ is either complete or anti-complete to $W_j \cap V_\ell$.
\end{enumerate1}
\item If $i=j$ and $k +1 \equiv \ell \; (\bmod\; 3)$ then:
\begin{enumerate1}
\item \label{prop:3-Ti-Wi-comp}$T^i \cap V_k$ is complete to $W_j \cap V_\ell$.
\end{enumerate1}
\item If $i=j$ and $k +1 \not\equiv \ell \; (\bmod\; 3)$ then:
\begin{enumerate1}
\item \label{prop:3-Ti-Wi-anti}$T^i \cap V_k$ is anti-complete to $W_j\cap V_\ell$.
\end{enumerate1}
\end{enumeratei}

Next, we show that basic graphs have bounded clique-width.

\begin{lemma}\label{lem:basic-bdd-cw}
If~$G$ is a basic graph then it has clique-width at most~$9$.
\end{lemma}
\begin{proof}
Let~$G$ be a graph with vertices partitioned into sets as above.
This means that we have sets of vertices $T^1,W_1,T^2,W_2,\ldots,T^p,W_p$ in order, such that if~$X$ and~$Y$ are sets in this order with~$X$ coming before~$Y$ then $X \cap V_k$ is complete to $Y \cap V_\ell$ if $k+1 \equiv \ell \; (\bmod\; 3)$ and anti-complete otherwise in all cases except where $X=W_i, Y=W_{i+1}$ for some~$i$, in which case $X \cap V_k$ may either be complete or anti-complete to $Y \cap V_{k+1}$.
Also recall that~$U^{T^i}$ is an independent set for every~$i$ and there is a vertex $x \in T^i$ such that every vertex of~$U^{T^i}$ has the same neighbourhood in $G \setminus T^i$ as~$x$.

Note that $W_i \subseteq V_1 \cup V_2 \cup V_3$.
Then~$G[W_i]$ is a $3$-partite graph with $3$-partition $(W_i\cap V_1,W_i \cap V_2,W_i \cap V_3)$.
Furthermore, $G[W_i]$ is $K_3$-free, and contains no induced~$3P_1$ with exactly one vertex in each~$V_j$.
Since~$G[W_i]$ is $(P_1+\nobreak 2P_2)$-free it must therefore be $(P_7,S_{1,2,3})$-free.
Therefore, by Lemma~\ref{lem:123-imp}, the graph~$G[W_i]$ is totally $3$-decomposable with respect to this partition.
By Lemma~\ref{lem:k-decomp}, we can construct~$G[W_i]$ using at most six labels such that the resulting labelled graph has all vertices in~$W_i$ labelled with label~$i$ for $i \in \{1,2,3\}$.

We are now ready to describe how to construct~$G$. We do this by constructing $G[T^i \cup U^{T^i}]$ then~$G[W_i]$ for each~$i \in \{1,\ldots,p\}$ in turn and adding it to the graph.
More formally, we start with the empty graph, then for $i=1,\ldots,p$ in turn, we do the following:

\begin{enumerate}
\item Let $\{x_1^i,x_2^i,x_3^i\}=T^i$, where $x_j^i \in V_j$ for $j \in \{1,2,3\}$.
Add vertices $x_1^i,x_2^i$ and~$x_3^i$ with labels $4,5$ and~$6$, respectively, then add edges between vertices labelled $4\&5, 5\&6$ and $4\&6$.

\item If~$U^{T^i}$ is non-empty then the vertices in this set have the same neighbourhood in $G\setminus T^i$ as $x_1^i,x_2^i$ or~$x_3^i$.
Add the vertices of~$U^{T^i}$ with label $4,5$ or~$6$, respectively.

\item Add edges between vertices labelled $1\&5, 2\&6$ and $3\&4$.

\item Relabel vertices labelled $4,5$ or~$6$ to have labels $1,2$ or~$3$, respectively.

\item Construct~$G[W_i]$ with vertices labelled $4,5$ or~$6$, if they are in $V_1,V_2$ or~$V_3$, respectively.

\item Add edges between vertices labelled $1\&5, 2\&6$ and $3\&4$.

\item If $i>1$ then add edges between vertices labelled $4\&9, 5\&7$ and $6\&8$ if $V_k \cap W_{i}$ is complete to $V_{k-1} \cap W_{i-1}$ for $k=1,2,3$, respectively.

\item Relabel vertices labelled $7,8$ or~$9$ to have labels $1,2$ or~$3$, respectively.

\item Relabel vertices labelled $4,5$ or~$6$ to have labels $7,8$ or~$9$, respectively.
\end{enumerate}

Note that at the end of any iteration of the above procedure, the vertices of~$W_i$ will have labels in $\{7,8,9\}$ and all other constructed vertices will have labels in $\{1,2,3\}$.

This construction builds a copy of~$G$ using at most nine labels.
Thus~$G$ has clique-width at most~$9$.
This concludes the proof of the lemma.
\qedllncs
\end{proof}

We are now ready to prove our main theorem of this section.
To do so, we show that if a graph~$G$ is $(\mbox{diamond},P_1+\nobreak 2P_2)$-free then either we can show that~$G$ has bounded clique-width directly (possibly by applying some graph operations that do not change the clique-width the graph by ``too much'') or else the (unmodified) graph~$G$ is itself basic (in which case it has clique-width at most~$9$).

\begin{theorem}\label{t-four}
The class of $(\mbox{diamond},P_1+\nobreak 2P_2)$-free graphs has bounded~clique-width.
\end{theorem}

\begin{proof}
Let~$G$ be a $(\mbox{diamond},P_1+\nobreak 2P_2)$-free graph.
By Lemma~\ref{lem:discon}, we may assume that~$G$ is connected.
By Theorem~\ref{thm:gen2}, we may assume that~$G$ contains an induced~$K_3$.
By Lemma~\ref{lem:K4}, we may assume that~$G$ is $K_4$-free.

Let~$T$ be an arbitrary induced triangle (i.e.~$K_3$) in~$G$ with vertices $v^T_1,v^T_2$ and~$v^T_3$.
Since~$G$ is $(\mbox{diamond},K_4)$-free, every vertex not in~$T$ has at most one neighbour in~$T$.
For $i \in \{1,2,3\}$ let~$V^T_i$ be the set of vertices not in~$T$ whose unique neighbour in~$T$ is~$v^T_i$ and let~$U^T$ be the set of vertices that have no neighbour in~$T$.
We will now prove a series of claims.
More formally, we will show that if the conditions of any of these claims are not satisfied, then either we obtain a contradiction or we can directly prove that~$G$ has bounded clique-width, in which case we are done.

\clm{\label{clm:3-Vi-large}For every triangle~$T$, the sets $V^T_1$, $V^T_2$ and~$V^T_3$ each contain at least three vertices.}
If for some~$i$ the set~$V^T_i$ contains at most two vertices then~$v^T_i$ has at most four neighbours in~$G$.
If we delete every vertex in~$N(v^T_i)$, then~$v^T_i$ has no neighbours in the resulting graph.
Therefore either~$G$ has at most five vertices (in which case it has clique-width at most~$5$), or $G \setminus N(v^T_i)$ is a disconnected $(\mbox{diamond},P_1+2P_2)$-free graph, so it has bounded clique-width by Lemma~\ref{lem:discon}.
By Fact~\ref{fact:del-vert}, it follows that~$G$ has bounded clique-width.
This completes the proof of the claim.

\clm{\label{clm:3-ViT-indep}For every triangle~$T$, the sets $V^T_1$, $V^T_2$ and~$V^T_3$ are independent.}
Suppose, for contradiction, that~$V^T_1$ is not an independent set.
Since~$G$ is $K_4$-free and every vertex of~$V^T_1$ is adjacent to~$v^T_1$, it follows that~$G[V^T_1]$ is $K_3$-free.
Since~$V^T_1$ contains at least three vertices by Claim~\ref{clm:3-Vi-large}, there must be vertices $x,y,z \in V^T_1$ such that~$x$ is adjacent to~$y$, but not to~$z$.
Then $G[v^T_1,y,x,z]$ is a diamond if~$y$ and~$z$ are adjacent and $G[z,x,y,v^T_2,v^T_3]$ is a $P_1+\nobreak 2P_2$ if they are not.
This contradiction implies that~$V^T_1$ is an independent set.
The claim follows by symmetry.

\clm{\label{clm:3-vertex-disjoint-triangles}Every pair of triangles in~$G$ is vertex-disjoint.}
Consider a triangle~$T$ with vertex~$v^T_1$.
The neighbourhood of~$v^T_1$ is $V^T_1 \cup \{v^T_2,v^T_3\}$.
Now~$V^T_1$ is independent by Claim~\ref{clm:3-ViT-indep} and anti-complete to $\{v^T_2,v^T_3\}$ by definition.
Therefore, if a triangle in~$G$ contains~$v^T_1$ then it must also contain~$v^T_2$ and~$v^T_3$.
In other words, $v^T_1$ is contained in only one triangle in~$G$, namely~$T$.
The claim follows by symmetry.

\clm{\label{clm:3-UT-indep}For every triangle~$T$, the set $U^T$ is independent.}
By Claim~\ref{clm:3-Vi-large}, we can choose $x,y \in V^T_1$ and by Claim~\ref{clm:3-ViT-indep}, $x$ must be non-adjacent to~$y$.
If a vertex $u \in U^T$ is adjacent to~$x$, but not to~$y$ then $G[y,u,x,v^T_2,v^T_3]$ is a $P_1+\nobreak 2P_2$, a contradiction.
Therefore every vertex of~$U^T$ is either complete or anti-complete to $\{x,y\}$.
Suppose $u,v \in U^T$.
First suppose~$u$ and~$v$ are non-adjacent.
If~$x$ is adjacent to~$u$ but~$v$ is not, then $G[v,u,x,v^T_2,v^T_3]$ is a $P_1+\nobreak 2P_2$, a contradiction.
Therefore if $u,v\in U^T$ are non-adjacent, then $\{u,v\}$ is either complete or anti-complete to $\{x,y\}$.
Now suppose~$u$ and~$v$ are adjacent.
Then $G[u,v,x,y]$ is a diamond if $\{u,v\}$ is complete to $\{x,y\}$ and $G[x,u,v,v^T_2,v^T_3]$ is a $P_1+\nobreak 2P_2$ if $\{u,v\}$ is anti-complete to $\{x,y\}$.
Therefore if $u,v \in U^T$ are adjacent then exactly one of them is complete to $\{x,y\}$ and the other is anti-complete to $\{x,y\}$.
This means that~$G[U^T]$ is a complete bipartite graph, with partition classes~$U^T_1$ and~$U^T_2$, say, and furthermore, one of~$U^T_1$ and~$U^T_2$ is complete to~$V^T_1$ and the other is anti-complete to~$V^T_1$.
Similarly, this holds with the same partition $(U^T_1,U^T_2)$ if we replace~$V^T_1$ by~$V^T_2$ or~$V^T_3$.
Thus every vertex of~$U^T_1$ (respectively~$U^T_2$) has the same neighbourhood in $V^T_1 \cup V^T_2 \cup V^T_3$.

Suppose that~$V^T_i$ and~$V^T_j$ are both complete to~$U^T_k$ for some $i,j \in \{1,2,3\}$ with $i \neq j$ and some $k \in \{1,2\}$ and that~$U^T_k$ contains at least two vertices, say~$u$ and~$v$.
If $x \in V^T_i$ and $y \in V^T_j$ are adjacent, then $G[x,y,u,v]$ is a diamond, a contradiction. Therefore~$V^T_i$ is anti-complete to~$V^T_j$.

Suppose that~$U^T_1$ and~$U^T_2$ each contain at least one vertex, say~$u$ and~$v$, respectively.
We will show that in this case the clique-width of~$G$ is bounded.
Suppose, for contradiction, that $G \setminus (T \cup \{u,v\})$ contains an induced~$K_3$, say with vertex set~$T'$.
Since~$G[U^T]$ is a complete bipartite graph with bipartition $(U^T_1,U^T_2)$ and no vertex of a set~$V^T_i$ can have neighbours in both~$U^T_1$ and~$U^T_2$, at most one vertex of~$T'$ can be in~$U^T$.
Suppose that~$U^T_1$ contains at least two vertices (so $U^T_1 \setminus \{u\}$ is non-empty) and that~$U^T_1$ is complete to~$V^T_i$ and~$V^T_j$ for some $i \neq j$ (in which case~$U^T_2$ is anti-complete to~$V^T_i$ and~$V^T_j$).
Then~$V^T_i$ and~$V^T_j$ must be anti-complete. We conclude that in this case no vertex of~$U^T_1$ can belong to~$T'$.
No vertex of~$U^T_2$ can belong to~$T'$ either, since vertices in~$U^T_2$ can only have neighbours in~$U^T_1$ and in~$V^T_k$ where $k \notin \{i,j\}$ (if~$U^T_1$ is anti-complete to~$V^T_k$).
Furthermore, since~$V^T_i$ is anti-complete to~$V^T_j$, and~$V^T_1$, $V^T_2$,~$V^T_3$ are independent (by Claim~\ref{clm:3-ViT-indep}), there is no induced~$K_3$ in $G[V^T_1 \cup V^T_2 \cup V^T_3]$.
Thus~$T'$ cannot exist, a contradiction.

The above means that if such a triangle~$T'$ does exist and a set~$U^T_i$ contains at least two vertices, then~$U^T_i$  must be anti-complete to at least two distinct sets~$V^T_j$ and~$V^T_k$ (in which case~$U^T_i$ cannot contain a vertex of~$T'$).
Since~$T'$ consists of vertices of $G \setminus (T \cup \{u,v\})$, this means that no vertex of~$U^T$ is in~$T'$ (if~$U^T_i$ contains a single vertex for some~$i$ then by definition~$T'$ does not include it).
By Claim~\ref{clm:3-ViT-indep}, it follows that~$T'$ must consist of vertices $x \in V^T_1, y \in V^T_2$ and~$z \in V^T_3$.
Since each set~$V^T_i$ is anti-complete to exactly one of~$U^T_1$ and~$U^T_2$, we may assume without loss of generality that~$U^T_1$ (and therefore~$u$) is complete to both~$V^T_1$ and~$V^T_2$.
Now $G[x,y,z,u]$ is a~$K_4$ or diamond if~$u$ and~$z$ are adjacent or non-adjacent, respectively.
This contradiction means that $G \setminus (T \cup \{u,v\})$ must in fact be~$K_3$-free.
Since $G \setminus (T \cup \{u,v\})$ is a $(K_3,P_1+\nobreak 2P_2)$-free graph, it has bounded clique-width by Theorem~\ref{thm:gen2}.
By Fact~\ref{fact:del-vert}, we conclude that~$G$ also has bounded clique-width.
We may therefore assume that either~$U^T_1$ or~$U^T_2$ is empty.
It follows that~$U^T$ is an independent set.
This completes the proof of the claim.

\clm{\label{clm:3-UT-nice-nbhd}For every triangle~$T$, there is a vertex~$x \in V(T)$ such that $N(x) = N(u) \cup (V(T) \setminus \{x\})$ for all $u \in U^T$.}
By the previous claim, we may assume that~$U^T$ is independent.
Note that by the same arguments as for the previous claim, for all $i \in \{1,2,3\}$, $U^T$ is trivial to $V^T_i$.
Suppose $u \in U^T$.
By the same arguments as for the previous claim, $U^T$ must be anti-complete to at least two distinct sets~$V^T_i$ and~$V^T_j$, otherwise $G \setminus (T \cup \{u\})$ would be $K_3$-free and the clique-width of~$G$ would be bounded as before.
Since~$G$ is connected, it follows that~$U^T$ must be complete to at least one set~$V^T_i$.
Therefore~$U^T$ must be complete to exactly one set~$V^T_i$.
It follows that $N(v^T_i) = V^T_i \cup (V(T) \setminus \{v^T_i\}) = N(u) \cup (V(T) \setminus \{v^T_i\})$ for all $u \in U^T$.
This completes the proof of the claim.

\clm{\label{clm:3-no-U-triangle}No triangle in~$G$ contains a vertex of~${\cal U}$.}
If $u \in {\cal U}$ then $u \in U_T$ for some triangle~$T$.
By the previous claim, the neighbourhood of every vertex of~$U^T$ is~$V^T_i$, for some~$i$.
Since~$V^T_i$ is an independent set, the claim follows immediately.

\clm{\label{clm:3-induced-mathings-between-triangles}If~$T$ and~$T'$ are distinct triangles in~$G$ then the edges between them form an induced matching.}
Suppose~$T$ and~$T'$ are distinct triangles in~$G$.
By Claim~\ref{clm:3-vertex-disjoint-triangles}, $T$ and~$T'$ must be vertex-disjoint.
By Claim~\ref{clm:3-no-U-triangle}, it follows that every vertex of~$T'$ is in $V^T_1 \cup V^T_2 \cup V^T_3$, so every vertex of~$T'$ has exactly one neighbour in~$T$.
By Claim~\ref{clm:3-ViT-indep}, for $i \in \{1,2,3\}$, the set~$V^T_i$ is an independent set, so it can contain at most one vertex of~$T'$.
Therefore~$T'$ has exactly one vertex in each of~$V^T_1$, $V^T_2$ and~$V^T_3$.
By definition of~$V^T_i$, this means that every vertex of~$T'$ has a different neighbour in~$T$.
The claim follows.

\clm{\label{clm:3-Vi-Vj-2P2-free}For every triangle~$T$ and for every pair of distinct $i,j \in \{1,2,3\}$, $G[V^T_i \cup V^T_j]$ is $2P_2$-free.}
Suppose, for contradiction, that $G[V^T_1 \cup V^T_2]$ contains an induced~$2P_2$.
Then this~$2P_2$, together with the vertex~$v^T_3$ would induce a $P_1+\nobreak 2P_2$ in~$G$.
The claim follows by symmetry.

\clm{\label{clm:3-no-3P_1}For every triangle~$T$, there is no induced~$3P_1$ in~$G$ with one vertex in each of~$V^T_1,V^T_2$ and~$V^T_3$.}
Suppose that there are three vertices $x \in V^T_1, y \in V^T_2$ and $z \in V^T_3$ that are pairwise non-adjacent.
We will show that in this case~$G$ has bounded clique-width.
Suppose $u \in U^T$.
By Claim~\ref{clm:3-UT-nice-nbhd}, $u$ has exactly one neighbour in $\{x,y,z\}$.
Without loss of generality, assume that~$u$ is adjacent to~$x$.
Then $G[z,u,x,y,v^T_2]$ is a $P_1+\nobreak 2P_2$, a contradiction.
We may therefore assume that~$U^T$ is empty.
If there is a vertex $x' \in V^T_1 \setminus \{x\}$ that is adjacent to~$y$, but not to~$z$ then $G[x,x',y,v^T_3,z]$ is a $P_1+\nobreak 2P_2$ in~$G$.
This contradiction means that every vertex of~$V^T_1$ is either complete or anti-complete to $\{y,z\}$.
Similarly, every vertex of~$V^T_2$ is either complete or anti-complete to $\{x,z\}$ and every vertex of~$V^T_3$ is either complete or anti-complete to $\{x,y\}$.
Note that the above holds for any three pairwise non-adjacent vertices in $V^T_1,V^T_2$ and~$V^T_3$, respectively.

Let~$V'^T_1$ and~$V''^T_1$ be the sets of vertices in~$V^T_1$ that are anti-complete or complete to $\{y,z\}$, respectively.
Let~$V'^T_2$ and~$V''^T_2$ be the sets of vertices in~$V^T_2$ that are anti-complete or complete to $\{x,z\}$, respectively.
Let~$V'^T_3$ and~$V''^T_3$ be the sets of vertices in~$V^T_3$ that are anti-complete or complete to $\{x,y\}$, respectively.
Note that $x \in V'^T_1,\allowbreak y \in\nobreak V'^T_2$ and $z \in V'^T_3$.

Suppose $x' \in V'^T_1$ and $y' \in V'^T_2$.
Since~$x'$ is non-adjacent to~$y$ and to~$z$, it follows that $G[x',y,z]$ is a~$3P_1$.
Since~$y'$ is non-adjacent to~$z$, it must therefore be anti-complete to $\{x',z\}$.
In particular, this means that if $i,j \in \{1,2,3\}$ are distinct then~$V'^T_i$ is anti-complete to~$V'^T_j$.

Suppose $x' \in V'^T_1$ and $y' \in V''^T_2$.
Since~$x'$ is non-adjacent to~$y$ and to~$z$, it follows that $G[x',y,z]$ is a~$3P_1$.
Since~$y'$ is adjacent to~$z$, it must therefore be complete to $\{x',z\}$.
In particular, this means that if $i,j \in \{1,2,3\}$ are distinct then~$V'^T_i$ is complete to~$V''^T_j$.

Note that for all $i \in \{1,2,3\}$, $V'^T_i$ is anti-complete to~$V''^T_i$, since~$V^T_i$ is an independent set.

Suppose $x' \in V''^T_1$ and $y' \in V''^T_2$.
If~$x'$ and~$y'$ are non-adjacent then $G[x',y,x,y']$ is a~$2P_2$ in $G[V^T_1 \cup V^T_2]$, which would contradict Claim~\ref{clm:3-Vi-Vj-2P2-free}.
This means that if $i,j \in \{1,2,3\}$ are distinct then~$V''^T_i$ is complete to~$V''^T_j$.

We now proceed as follows: from~$G$, we delete the three vertices of~$T$.
We then apply a bipartite complementation between every pair of sets~$V'^T_i$ and~$V''^T_j$ and every pair of distinct sets~$V''^T_i$ and~$V''^T_j$ (a total of nine bipartite complementations).
After doing this, we obtain an edge-less graph, which therefore has clique-width at most~$1$.
By Facts~\ref{fact:del-vert} and~$\ref{fact:bip}$, it follows that~$G$ must also have bounded clique-width.
This completes the proof of the claim.

\clm{\label{clm:3-p-geq-3}$G$ contains at least three vertex-disjoint triangles.}
Suppose, for contradiction, that the claim is false.
Then~$G$ contains at most two vertex-disjoint triangles, in which case, we can delete at most six vertices to obtain a $(K_3,P_1+\nobreak 2P_2)$-free graph, which has bounded clique-width by Theorem~\ref{thm:gen2}.
By Fact~\ref{fact:del-vert}, $G$ also has bounded clique-width.
This completes the proof of the claim.

\bigskip
We will now assume that the above claims are satisfied and show that this implies that~$G$ is basic.
We arbitrarily fix a triangle~$T^1$ with vertices~$v^{T^1}_1$, $v^{T^1}_2$ and~$v^{T^1}_3$.
To simplify notation, set $v_i=v^{T^1}_i$ for $i \in \{1,2,3\}$.
Recall that by Claim~\ref{clm:3-no-U-triangle}, no~$K_3$ in~$G$ has a vertex in~${\cal U}$.
By Claim~\ref{clm:3-ViT-indep}, it follows that every~$K_3$ in~$G$ apart from~$T^1$ has exactly one vertex in each of $V^{T^1}_1 \setminus {\cal U}$, $V^{T^1}_2 \setminus {\cal U}$ and~$V^{T^1}_3 \setminus {\cal U}$.
We now set $V_1=(V^{T^1}_1 \cup \{v_2\})\setminus {\cal U}$, $V_2=(V^{T^1}_2 \cup \{v_3\}) \setminus {\cal U}$ and $V_3=(V^{T^1}_3 \cup \{v_1\}) \setminus {\cal U}$.

\clm{\label{clm:3-Vi-indep}$V_1,V_2$ and~$V_3$ are independent.}
The vertices in~$V^{T^1}_i$ are exactly the vertices outside~$T^1$ whose unique neighbour in~$T^1$ is~$v_i$.
The claim follows by Claim~\ref{clm:3-ViT-indep}.

\medskip
By Claim~\ref{clm:3-vertex-disjoint-triangles} any two triangles in~$G$ must be vertex-disjoint.
By Claim~\ref{clm:3-induced-mathings-between-triangles}, the edges between any two triangles in~$G$ form a perfect matching.
Let $T^x = \{x_1,x_2,x_3\}$ and $T^y = \{y_1,y_2,y_3\}$ be two distinct triangles in~$G$ with $x_i,y_i \in V_i$ for $i \in \{1,2,3\}$.
By Claim~\ref{clm:3-Vi-indep}, $x_i$ is non-adjacent to~$y_i$ for $i \in \{1,2,3\}$.
This means that the set of edges between~$T^x$ and~$T^y$ is either $\{x_1y_2,x_2y_3,x_3y_1\}$ or $\{x_1y_3,x_2y_1,x_3y_2\}$.
We say that $T^x < T^y$ holds in the first case and $T^y < T^x$ holds in the second.
Note that exactly one of these statements holds for any two distinct triangles in $G$.
Furthermore, note that if~$T^x$ is a triangle other than~$T^1$ then the definition of the sets~$V_i$ implies that $T^1 < T^x$.

We show that the relation~$<$ is transitive.
Suppose, for contradiction, that this is not the case.
Then there must be three pairwise distinct triangles in~$G$, say $T^x = \{x_1,x_2,x_3\}, T^y = \{y_1,y_2,y_3\}$ and $T^z = \{z_1,z_2,z_3\}$, where $x_i,y_i,z_i \in V^T_i$ for $i \in \{1,2,3\}$, with $T^x < T^y$, $T^y < T^z$ and $T^z < T^x$.
Then~$x_1$ is adjacent to~$y_2$, $y_2$ is adjacent to~$z_3$ and~$z_3$ is adjacent to~$x_1$.
Therefore $G[x_1,y_2,z_3]$ is a~$K_3$ which shares exactly one vertex with~$T^x$, which would contradict Claim~\ref{clm:3-vertex-disjoint-triangles}.
Therefore~$<$ is a transitive, anti-symmetric relation on the triangles in~$G$.
We may now order the triangles in~$G$, say $T^1 < T^2 < \cdots < T^p$ for some~$p$.
By Claim~\ref{clm:3-p-geq-3}, it follows that $p \geq 3$.
We now conclude the following:

\clm{\label{clm:3-have-all-K3s}$\{G[T^1],\ldots,G[T^p]\}$ is the set of all induced triangles in~$G$ and each of them has exactly one vertex in each of $V_1,V_2$ and~$V_3$.}
\smallclm{\label{clm:3-Ti-Tj-anti}If $i<j$ and $k +1 \not\equiv \ell \; (\bmod\; 3)$ then $T^i \cap V_k$ is anti-complete to $T^j\cap V_\ell$.}
\lastsmallclm{\label{clm:3-Ti-Tj-comp}If $i<j$ and $k +1 \equiv \ell \; (\bmod\; 3)$ then $T^i \cap V_k$ is complete to $T^j \cap V_\ell$.}

\medskip
Consider a vertex~$x$ that is not in any induced triangle in~$G$.
If $x \notin {\cal U}$ then $x \in V_1 \cup V_2 \cup V_3$ and~$x$ must have exactly one neighbour in every triangle in~$G$.
Let~$W$ be the set of vertices that are not in any triangle in~$G$ and have exactly one neighbour in every induced triangle in~$G$.

We extend the relation~$<$ as follows:
suppose~$T=\{x_1,x_2,x_3\}$ is an induced triangle in~$G$ with $x_1 \in V_1, x_2 \in V_2$ and $x_3 \in V_3$ and suppose $w \in W$.
Then~$w$ is a vertex in~$V_i$ for some $i \in \{1,2,3\}$.
By Claim~\ref{clm:3-Vi-indep}, $w$ is not adjacent to~$x_i$.
Since $w \in W$, $w$ must be adjacent to exactly one vertex of~$T$.
We say that $x < T$ holds if~$x$ is adjacent to~$x_{i+1}$ and $T < x$ if~$x$ is adjacent to~$x_{i-1}$ (we interpret indices modulo~$3$).

Let $w \in W$ and let~$T$ and~$T'$ be triangles in~$G$ such that $w < T$ and $T < T'$.
We will show that $w < T'$.
Say $T=\{x_1,x_2,x_3\}$ and $T'=\{y_1,y_2,y_3\}$, where $x_i,y_i \in V_i$ for $i \in \{1,2,3\}$.
Without loss of generality, assume $w \in V_1$.
Since $w < T$, $w$ is adjacent to~$x_2$.
Since $T < T'$, $x_2$ is adjacent to~$y_3$.
Since $w \in V_1$, $w$ is non-adjacent to~$y_1$.
Now~$w$ cannot be adjacent to~$y_3$, otherwise~$G[w,x_2,y_3]$ would be a triangle that is not vertex-disjoint from~$T$, which would contradict Claim~\ref{clm:3-vertex-disjoint-triangles}.
Since $w \in W$, it must have a neighbour in~$T'$, so~$w$ must therefore be adjacent to~$y_2$.
It follows that $w < T'$.
Similarly, if $T < T'$ and $T' < w$ then $T < w$ and if $T < w$ and $w < T'$ then $T < T'$.

This means that we can now partition~$W$ into sets $W_1,\ldots,W_p$ where~$W_i$ contains the vertices $x \in W$ such that $T^j < x$ for $j \leq i$ and $x < T^j$ for $j >i$.
(Note that $T^1<w$ for all $w \in W$, by construction.)
We immediately conclude the following:

\clm{\label{clm:3-Ti-Wj-anti}If $i<j$ and $k +1 \not\equiv \ell \; (\bmod\; 3)$ then $T^i \cap V_k$ is anti-complete to $W_j\cap V_\ell$.}
\smallclm{\label{clm:3-Ti-Wi-anti}If $i=j$ and $k +1 \not\equiv \ell \; (\bmod\; 3)$ then $T^i \cap V_k$ is anti-complete to $W_j\cap V_\ell$.}
\smallclm{\label{clm:3-Wi-Tj-anti}If $i<j$ and $k +1 \not\equiv \ell \; (\bmod\; 3)$ then $W_i \cap V_k$ is anti-complete to $T^j\cap V_\ell$.}
\smallclm{\label{clm:3-Ti-Wj-comp}If $i<j$ and $k +1 \equiv \ell \; (\bmod\; 3)$ then $T^i \cap V_k$ is complete to $W_j \cap V_\ell$.}
\smallclm{\label{clm:3-Ti-Wi-comp}If $i=j$ and $k +1 \equiv \ell \; (\bmod\; 3)$ then $T^i \cap V_k$ is complete to $W_j \cap V_\ell$.}
\lastsmallclm{\label{clm:3-Wi-Tj-comp}If $i<j$ and $k +1 \equiv \ell \; (\bmod\; 3)$ then $W_i \cap V_k$ is complete to $T^j \cap V_\ell$.}

\medskip
We also prove the following claim:

\clm{\label{clm:3-GW_i-nice}$G[W_i]$ is $(P_1+\nobreak 2P_2)$-free and does not contain an induced~$3P_1$ with one vertex in each of $V_1,V_2$ and~$V_3$.}
Since~$G$ is $(P_1+\nobreak 2P_2)$-free, it follows that~$G[W_i]$ is also $(P_1+\nobreak 2P_2)$-free.
Since the vertices of~$W_i$ do not belong to any triangle of~$G$ and do not belong to ${\cal U}$, it follows that $W_i \subseteq V^{T^1}_1 \cup V^{T^1}_2 \cup V^{T^1}_3$.
The claim then follows by Claim~\ref{clm:3-no-3P_1}.

\medskip
It remains to analyse the edges between the sets $W_1,\ldots,W_p$.

\clm{\label{clm:3-Wi-Wj-anti}If $i<j$ and $k +1 \not\equiv \ell \; (\bmod\; 3)$ then $W_i \cap V_k$ is anti-complete to $W_j\cap V_\ell$.}
Let $i,j \in \{1,\ldots,p\}$ be such that $i<j$.
Let $T^j=\{x_1,x_2,x_3\}$ with $x_k \in V_k$ for $k \in \{1,2,3\}$.
Note that if $x \in W_i$ and $y \in W_j$ then $x < T^j$ and $T^j < y$.
Now $W_i \cap V_k$ is anti-complete to $W_j \cap V_k$ for $k \in \{1,2,3\}$, since~$V_k$ is an independent set by Claim~\ref{clm:3-Vi-indep}.
Suppose $x \in W_i \cap V_1$ and $y \in W_j \cap V_3$.
Then~$x$ and~$y$ are both adjacent to~$x_2$.
Therefore~$x$ and~$y$ cannot be adjacent, otherwise $G[x_2,x,y]$ would be a triangle which is not vertex-disjoint from~$T^j$, which would contradict Claim~\ref{clm:3-vertex-disjoint-triangles}.
By symmetry we conclude that $W_i \cap V_k$ is anti-complete to $W_j \cap V_{k+2}$ for $k \in \{1,2,3\}$ (interpreting subscripts modulo~$3$).
This completes the proof of the claim.

\medskip
The edges between $W_i \cap V_k$ and $W_j \cap V_{k+1}$ for $k \in \{1,2,3\}$ are more complicated, as shown in the following two claims:

\clm{\label{clm:3-Wi-Wj-comp}If $i+1<j$ and $k +1 \equiv \ell \; (\bmod\; 3)$ then $W_i \cap V_k$ is complete to $W_j \cap V_\ell$.}
Let $i,j \in \{1,\ldots,p\}$ be such that $i+1<j$.
Suppose, for contradiction, that $x \in W_i \cap V_1$ and $y \in W_j \cap V_2$ are non-adjacent.
Since $i+2\leq j$ we find that $x < T^{j-1},\allowbreak x<\nobreak T^j,\allowbreak T^{j-1}<y$ and $T^j <y$.
Let $T^j=\{x_1,x_2,x_3\}$ with $x_k \in V_k$ for $k \in \{1,2,3\}$.
Let $T^{j-1}=\{y_1,y_2,y_3\}$, where $y_k \in V_k$ for $k \in \{1,2,3\}$.
Then~$x$ is adjacent to~$y_2$, but non-adjacent to~$x_1$, while~$y$ is adjacent to~$x_1$, but non-adjacent to~$y_2$.
Since $T^{j-1} < T^j$ it follows that~$y_2$ is non-adjacent to~$x_1$.
Since $T^1 < x,y$, the vertex~$v_3$ must be non-adjacent to~$x$ and~$y$ (recall that $v_3=v^{T^1}_3$ and that this vertex has no neighbours in~$V_1$ or~$V_2$ apart from~$v_1$ and~$v_2$).
Now $G[v_3,x,y_2,x_1,y]$ is a $P_1+\nobreak 2P_2$, a contradiction.
By symmetry this completes the proof of the claim.

\clm{\label{clm:3-Wi-Wi+1-triv}If $i+1=j$ and $k+1 \equiv \ell \; (\bmod\; 3)$ then $W_i \cap V_k$ is either complete or anti-complete to $W_j \cap V_\ell$.}
Let $i,j \in \{1,\ldots,p\}$ with $i+1=j$.
Let $T^j=\{x_1,x_2,x_3\}$ with $x_k \in V_k$ for $k \in \{1,2,3\}$.
Assume, for contradiction, that the vertex sets $W_i \cap V_k$ and $W_j \cap V_{k+1}$ are not trivial to each-other for some $k \in \{1,2,3\}$.
Without loss of generality, we may assume that there is a vertex~$x$ with a neighbour~$y$ and a non-neighbour~$y'$ such that either $x \in W_i \cap V_1$ and $y,y' \in W_j \cap V_2$ or $y,y' \in W_i \cap V_1$ and $x \in W_j \cap V_2$.
Note that~$x_3$ is non-adjacent to $x,y$ and~$y'$.
Since $T^1 < x,y,y'$, the vertex~$v_3$ must be non-adjacent to~$x$ and~$y$ (recall that $v_3=v^T_3$ and that this vertex has no neighbours in~$V_1$ or~$V_2$ apart from~$v_1$ and~$v_2$).
Now $G[y',x,y,v_3,x_3]$ is a $P_1+\nobreak 2P_2$, a contradiction.
By symmetry this completes the proof of the claim.

\bigskip
The claims proved above imply all the necessary properties for~$G$ to be basic.
Indeed, Claim~\ref{clm:3-no-U-triangle} implies Property~\ref{prop:3-no-U-triangle} and Claims~\ref{clm:3-UT-indep} and~\ref{clm:3-UT-nice-nbhd} imply Property~\ref{prop:3-UT-nbhd}.
Claims \ref{clm:3-Vi-indep},
\ref{clm:3-have-all-K3s},
\ref{clm:3-GW_i-nice},
\ref{clm:3-Ti-Tj-anti},
\ref{clm:3-Ti-Wj-anti},
\ref{clm:3-Wi-Tj-anti},
\ref{clm:3-Wi-Wj-anti},
\ref{clm:3-Ti-Tj-comp},
\ref{clm:3-Ti-Wj-comp},
\ref{clm:3-Wi-Tj-comp},
\ref{clm:3-Wi-Wj-comp},
\ref{clm:3-Wi-Wi+1-triv},
\ref{clm:3-Ti-Wi-comp} and
\ref{clm:3-Ti-Wi-anti}
imply Properties 
\ref{prop:3-Vi-indep},
\ref{prop:3-have-all-K3s},
\ref{prop:3-GW_i-nice},
\ref{prop:3-Ti-Tj-anti},
\ref{prop:3-Ti-Wj-anti},
\ref{prop:3-Wi-Tj-anti},
\ref{prop:3-Wi-Wj-anti},
\ref{prop:3-Ti-Tj-comp},
\ref{prop:3-Ti-Wj-comp},
\ref{prop:3-Wi-Tj-comp},
\ref{prop:3-Wi-Wj-comp},
\ref{prop:3-Wi-Wi+1-triv},
\ref{prop:3-Ti-Wi-comp},
\ref{prop:3-Ti-Wi-anti} and
 respectively.
Therefore~$G$ is basic, so it has bounded clique-width by Lemma~\ref{lem:basic-bdd-cw}.
This completes the proof.
\qedllncs
\end{proof}

\bibliography{mybib}

\end{document}